\documentclass[10pt,journal,letterpaper,compsoc]{IEEEtran}
\usepackage{epstopdf}
\usepackage{graphicx}
\usepackage{stmaryrd}
\usepackage{amssymb,amsmath}
\usepackage{multirow,url}
\usepackage{color,cite}

\newtheorem{theorem}{Theorem}
\newtheorem{observation}{Observation}

\ifodd 0
\newcommand{\rev}[1]{{\color{blue}#1}} 
\newcommand{\com}[1]{\textbf{\color{red}(COMMENT: #1)}} 
\newcommand{\comg}[1]{\textbf{\color{green}(COMMENT: #1)}}
\newcommand{\response}[1]{\textbf{\color{magenta}(RESPONSE: #1)}} 
\else
\newcommand{\rev}[1]{#1}
\newcommand{\com}[1]{}
\newcommand{\comg}[1]{}
\newcommand{\response}[1]{}
\fi

\begin{document}
%
\title{Investment and Pricing with Spectrum Uncertainty: A Cognitive Operator's Perspective}
%
%
%
%

\author{Lingjie~Duan,~\IEEEmembership{Student Member,~IEEE,}
        Jianwei~Huang,~\IEEEmembership{Member,~IEEE,}
        and~Biying~Shou
\IEEEcompsocitemizethanks{\IEEEcompsocthanksitem Lingjie Duan and
Jianwei Huang are with the Department of Information Engineering,
The Chinese University of Hong Kong,
Hong Kong.\protect\\
E-mail: \{dlj008, jwhuang\}@ie.cuhk.edu.hk. \IEEEcompsocthanksitem
Biying Shou is with the Department of Management Sciences, City
University of Hong Kong, Hong Kong.
E-mail: biying.shou@cityu.edu.hk.}
\thanks{Part of the results has appeared in IEEE INFOCOM, San Diego, USA, March 2010 \cite{DuanHuangShou2010}.}}



\IEEEcompsoctitleabstractindextext{%
\begin{abstract}
This paper studies the optimal investment and pricing decisions of a
cognitive mobile virtual network operator (C-MVNO) under
\emph{spectrum supply uncertainty}. Compared with a traditional MVNO
who often leases spectrum via long-term contracts, a C-MVNO can
acquire spectrum dynamically in short-term by both \emph{sensing}
the empty ``spectrum holes'' of licensed bands and \emph{dynamically
leasing} from the spectrum owner. As a result, a C-MVNO can make
flexible investment and pricing decisions to match the current
demands of the secondary unlicensed users. Compared to dynamic
spectrum leasing, spectrum sensing is typically cheaper, but the
obtained useful spectrum amount is random due to primary licensed
users' stochastic traffic. The C-MVNO needs to determine the optimal
amounts of spectrum sensing and leasing by evaluating the trade off
between cost and uncertainty. The C-MVNO also needs to determine the
optimal price to sell the spectrum to the secondary unlicensed
users, taking into account wireless heterogeneity of users such as
different maximum transmission power levels and channel gains. We
model and analyze the interactions between the C-MVNO and secondary
unlicensed users as a Stackelberg game. We show several interesting
properties of the network equilibrium, including threshold
structures of the optimal investment and pricing decisions, the
independence of the optimal price on users' wireless
characteristics, and guaranteed fair and predictable \com{spectrum
allocations}QoS among users. We prove that these properties hold for
general SNR regime and general continuous distributions of sensing
uncertainty. We show that spectrum sensing can significantly improve
the C-MVNO's expected profit and users' payoffs.
\end{abstract}

\begin{keywords}
\rev{Cognitive radio, spectrum trading, spectrum sensing, dynamic
spectrum leasing, spectrum pricing, Stackelberg game, Subgame Perfect
equilibrium.}
\end{keywords}}

\maketitle

\IEEEdisplaynotcompsoctitleabstractindextext

%
\IEEEpeerreviewmaketitle

\section{Introduction}
\IEEEPARstart{W}{ireless} spectrum is typically considered as a
scarce resource, and is traditionally allocated through static
licensing. Field measurements show that, however, most spectrum
bands are often under-utilized even in densely populated urban areas
(\cite{sharedspectrum2005}). To achieve more efficient spectrum
utilization, people have proposed various dynamic spectrum access
approaches  including
hierarchical-access and dynamic exclusive use
(\cite{Bahl,IEEEhowto:Zhao,Jayaweera,Buddhikot,Faulhaber}).
Hierarchical-access allows a secondary (unlicensed) network operator
or users to opportunistically access the spectrum without affecting
the normal operation of the spectrum owner who serves the primary
(licensed) users. Dynamic exclusive use allows a spectrum owner to
dynamically transfer and trade the usage right of its licensed
spectrum to a third party (e.g., a secondary network operator or a
secondary end-user) in the spectrum market. {This paper considers a
secondary operator who obtains spectrum resource via both
\emph{spectrum sensing} as in the hierarchical-access approach and
\emph{dynamic spectrum leasing} as in the dynamic exclusive use
approach.}

Spectrum sensing obtains awareness of the spectrum usage and
existence of primary users, by using geolocation and database,
beacons, or cognitive radios (e.g.,
\cite{Wellens,Tu,Ganesan,Yucek}).  The primary users are oblivious
to the presence of secondary cognitive network operators or users.
The secondary network operator or users can sense and utilize the
unused ``spectrum holes'' in the licensed spectrum without violating
the usage rights of the primary users (e.g.,
\cite{IEEEhowto:Zhao,Faulhaber}). Since the secondary operator or
users does not know the primary users' activities before sensing,
the amount of useful spectrum obtained through sensing is uncertain
(e.g.
\cite{Willkomm08,Willkomm09,huang2009optimal,huang2008opportunistic}).

With dynamic spectrum leasing, a spectrum owner allows secondary
users to operate in their temporarily unused part of spectrum in
exchange of economic return (e.g.,
\cite{Simeone,Jayaweera,Faulhaber}). The dynamic spectrum leasing
can be short-term or even real-time (e.g.,
\cite{Chapin,Chapin2,IEEEhowto:Jia}), and can be at a similar time
scale of the spectrum sensing operation.

In this paper, we study the operation of a cognitive radio network
that consists a cognitive mobile virtual network operator (C-MVNO)
and a group of secondary unlicensed users. The word ``virtual''
refers to the fact that the operator does not own the wireless
spectrum bands or even the physical network infrastructure. The
C-MVNO serves as the interface between the spectrum owner and the
secondary end-users. {The word ``cognitive'' refers to the fact that
the operator} can obtain spectrum resource through both spectrum
sensing using the cognitive radio technology and dynamic spectrum
leasing from the spectrum owner. The operator then resells the
obtained spectrum (bandwidth) to secondary users to maximize its
profit. {The proposed model is a hybrid of the hierarchical-access
and dynamic exclusive use models. It is applicable in various
network scenarios, such as achieving efficient utilization of the TV
spectrum in IEEE 802.22 standard\cite{stevenson2009ieee}. This
standard suggests that the secondary system should operate on a
point-to-multipoint basis, i.e., the communications will happen
between secondary base stations and secondary customer-premises
equipment. The base stations can be operated by one or several
C-MVNOs introduced in this paper.}

Compared with a  traditional MVNO who only leases spectrum through
long-term contracts, a C-MVNO can dynamically adjust its sensing and
leasing decisions to match the changes of users' demand at a short
time scale. {Moreover, sensing often offers a cheaper way to obtain
spectrum compared with leasing. The cost of sensing mainly includes
the sensing time and energy, and does not include explicit cost paid
to the spectrum owner. With a mature spectrum sensing technology,
sensing cost should be reasonable low (otherwise there is no point
of using cognitive radio). Spectrum leasing, however, involves
direct negotiation with the spectrum owner. When the spectrum owner
determines the cost of leasing, it needs to calculate its
opportunity cost, i.e., how much revenue the spectrum can provide if
the spectrum owner provides services directly over it. It is
reasonable to believe that the leasing cost is more expensive than
the sensing cost in most cases\footnote{{The analysis of this paper
also covers the case where sensing is more expensive than leasing,
which is a trivial case to study.}}. Although sensing is cheaper,
the amount of spectrum obtained through sensing is often uncertain
due to the stochastic nature of primary users' traffic.} It is thus
critical for a C-MVNO to find the right balance between cost and
uncertainty.

Our key results and contributions are summarized as follows. For
simplicity, we refer to the C-MVNO as ``operator'', secondary users
as ``users'', and ``dynamic leasing'' as ``leasing''.

\begin{itemize}


\item{\emph{A \com{dynamic}\rev{Stackelberg} game model}}: We model and analyze the interactions
between the operator and the users in the spectrum market as a
{Stackelberg game}. \rev{As the leader,} the operator makes
the sensing, leasing, and pricing decisions sequentially. \rev{As
the followers,} users then purchase bandwidth from the operator to
maximize their payoffs. \com{The users' bandwidth
demands payoff functions incorporate the heterogeneity of maximum
transmission power levels and channels gains.}\rev{By using backward
induction, we prove the existence and uniqueness of the equilibrium, and show how various system parameters (i.e., sensing and leasing costs, users' transmission power and channel conditions) affect the equilibrium behavior.} \com{Despite the complexity
of the model, we are able to fully characterize the unique
equilibrium behaviors of the operator and users.}

\item \emph{Threshold structures of the optimal investment and pricing decisions}:
At the equilibrium, the operator will sense the spectrum only if the
sensing cost is cheaper than a threshold. Furthermore, it will lease
some spectrum only if the resource obtained through sensing is
below a threshold. Finally, the operator will charge a constant
price to the users if the total bandwidth obtained through
sensing and leasing does not exceed a threshold. \rev{The thresholds are easy to compute and the corresponding decisions rules are easy to implement in practice.}


\item \emph{Fair and predictable \com{resource allocation}\rev{QoS}}:
The operator's optimal pricing decision is independent of the users'
wireless characteristics. Each user receives a \com{bandwidth
allocation}\rev{payoff} that is proportional to its channel gain and
transmission power, which leads to the same signal-to-noise (SNR)
for all users.


\item \emph{Impact of spectrum sensing}: We show that the availability of
sensing always increases the operator's profit in the
\emph{expected} sense. {The actual realization of the profit at a
particular time heavily depends on the spectrum sensing results.}
Users always get better payoffs when the operator performs spectrum
sensing.
\end{itemize}

Section \ref{sec:NetworkModel} introduces the network model and
problem formulation. In Section \ref{sec:BackwardInduction}, we
analyze the game model through backward induction. We discuss
various insights obtained from the equilibrium analysis and present
some numerical results in Section \ref{sec:Equilibrium}. In Section
\ref{sec:SensingImpact}, we show the impact of spectrum  sensing on
both the operator and users. We conclude in Section
\ref{sec:conclusion} and outline some future research directions.

\subsection{Related Work}
There is a growing interest in studying the investment and pricing
decisions of cognitive network operators recently. Several auction
mechanisms have been proposed  to study the investment problems of
cognitive network operators (e.g.,
\cite{sengupta2007economic,jia2009revenue}). Other recent results
studied the pricing decisions of the cognitive network operators who
interact with a group of secondary users (e.g.,
{\cite{manshaei-spectrum,IEEEhowto:Niyato,Inaltekin,IEEEhowto:Ileri,IEEEhowto:Xing,zhang2009stackelberg,niyato2007hierarchical,niyato2009dynamics}}).
\cite{sengupta2007economic} considered  users' queueing delays and
obtained most results through simulations. \cite{manshaei-spectrum}
presented a recent survey on the spectrum sharing games of network
operators and cognitive radio networks. \cite{IEEEhowto:Niyato}
studied the competition among multiple service providers without
modeling users' wireless details. \cite{Inaltekin} considered a
pricing competition game of two operators and adopted a simplified
wireless model for the users. \cite{IEEEhowto:Ileri} derived users'
demand functions based on the acceptance probability model for the
users. \cite{IEEEhowto:Xing} explored demand functions based on both
quality-sensitive and price-sensitive buyer population models.
{\cite{zhang2009stackelberg} formulated the interaction between one
primary user (monopolist) and multiple secondary users as a
Stackelberg game. The primary user uses some secondary users as
relays and leases its bandwidth to those relays to collect revenue.
\cite{niyato2007hierarchical} studied a multiple-level spectrum
market among primary, secondary, and tertiary services where global
information is not available. \cite{niyato2009dynamics} considered
the short-term spectrum trading between multiple primary users and
multiple secondary users. The spectrum buying behaviors of secondary
users are modeled as an evolutionary game, while selling behaviors
of primary users are modeled as a noncooperative game.
\cite{IEEEhowto:Ileri,IEEEhowto:Xing,zhang2009stackelberg,niyato2007hierarchical,niyato2009dynamics}
obtained most interesting results through simulations. There are
only few papers (e.g.,
\cite{IEEEhowto:Jia,Jia2,niyato2007hierarchical}) that jointly
considered the spectrum investment and service pricing problem as
this paper.  None of the above work considered the impact of supply
uncertainty due to spectrum sensing.}

%

Our  model of spectrum uncertainty  is related to the random-yield
model in supply chain management (e.g., \cite{IEEEhowto:Babich,
IEEEhowto:Sarang,Chain}).
{The unique wireless aspects of the system model lead to new
solutions and insights  in our problem.}

{Our paper represents a first attempt of understanding how spectrum
uncertainty impacts the economic decisions of an cognitive radio
operator. To obtain sharp insights, we focus on a stylized model
where a monopolist operator faces a group of secondary users. There
are many more interesting research issues in this area. Some are
further discussed in Section \ref{sec:conclusion}.}

\begin{figure}[tt]
\centering
\includegraphics[width=0.5\textwidth]{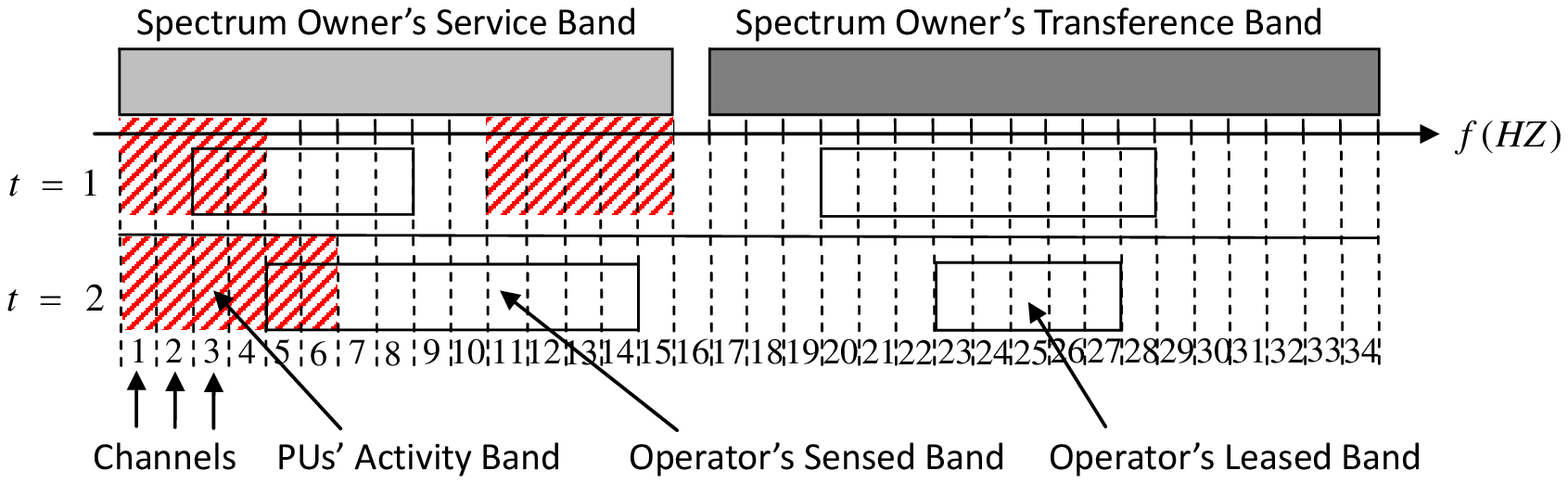}
\caption{Operator's Investment in Spectrum Sensing and Leasing}
\label{fig:spectrum pool}
\end{figure}

\section{Network Model}
\label{sec:NetworkModel}

\subsection{Background on Spectrum Sensing and Leasing}
To illustrate the opportunity and trade-off of spectrum sensing and
leasing, we consider a spectrum owner who divides its licensed
spectrum into two types:

\begin{itemize}
\item \emph{Service Band}: This band is reserved for serving the spectrum
owner's primary users (PUs). Since the PUs' traffic is stochastic,
there will be some unused spectrum which changes dynamically. The
operator can sense and utilize the unused portions. There are no
explicit communications between the spectrum owner and the operator.

\item \emph{Transference Band}: The spectrum owner temporarily does not
use this band. The operator can lease the bandwidth through explicit
communications with the spectrum owner. No sensing is allowed in
this band.
\end{itemize}

Due to the short-term property of both sensing and leasing, the
operator needs to make both the sensing and leasing decisions in
each time slot.

The example in Fig.~1 demonstrates the dynamic opportunities for
spectrum sensing, the uncertainty of sensing outcome, and the impact
of sensing or leasing decisions. The spectrum owner's entire band is
divided into small 34 channels\footnote{Channel 16 is the guard band
between the service and transference bands.}.
\begin{itemize}
\item Time slot 1:  PUs use
channels $1-4$ and $11-15$. The operator is unaware of this and
senses channels $3-8$. As a result, it obtains 4 unused channels
($5-8$). It leases additional 9 channels ($20-28$) from the
transference band.

\item Time slot 2:  PUs change their behavior and use channels $1-6$. The operator
senses channels $5-14$ and obtains 8 unused channels ($7-14$). It
leases additional 5 channels ($23-27$) from the transference band.
\end{itemize}

In this paper, we will only study the operator's decisions within a
single time slot. {We choose the time slot length such that primary
users' activities remain roughly unchanged within a single time
slot. This means that it is enough for the operator to sense at the
beginning of each time slot. For traffic types such as TV programs,
data transfer, and even VoIP voice sessions, the length of the time
slot can be reasonable long. For readers who are interested in the
optimization of the time slot length to balance sensing and data
transmission, see\cite{huang2009optimal}.}

\subsection{Notations and Assumptions}


\begin{table}
\centering \caption{Key Notations}
\begin{tabular}{|c|c|}
\hline Symbol & Physical Meaning
\\\hline \hline $B_s$ & Sensing
bandwidth\\\hline $B_l$ & Leasing bandwidth \\\hline $C_s$ & Unit
sensing cost\\\hline $C_l$ & Unit leasing cost\\\hline
$\alpha\in[0,1]$ & Sensing realization factor\\\hline
$\mathcal{I}=\{1,\cdots,I\}$ & Set of secondary users\\\hline $\pi$
& Unit price \\\hline $w_i$ & User $i$'s bandwidth allocation
\\\hline $r_i$ & User $i$'s data rate \\\hline $P_i^{\max}$ & User $i$'s maximum transmission power \\\hline
$h_i$ & User $i$'s channel gain \\\hline $n_0$ & Noise power density
\\\hline $g_i=P_i^{\max}h_i/n_0$ & User $i$'s wireless
characteristic \\\hline $\mathtt{SNR}_i=g_i/w_i$ & User $i$'s SNR
\\\hline $G=\sum_{i\in\mathcal{I}}g_i$ & Users' aggregate wireless
characteristics \\\hline $R$ & Operator's profit
\\\hline
\end{tabular}
\tabcolsep 5mm \label{tab:notation}
\end{table}

{We consider a cognitive network with one operator and a set
$\mathcal{I}=\{1,\ldots,I\}$ of users. The operator has the
cognitive capability and can sense the unused spectrum. One way to
realize this is to let the operator construct a sensor network that
is dedicated to sensing the radio environment in space and time
\cite{WeissDySPAN2010}. The operator will collect the sensing
information from the sensor network and provide it to the unlicensed
users, or providing ``sensing as service''. If the operator owns
several base stations, then each base station is responsible for
collecting sensing information in a certain geographical area. As
mentioned in \cite{WeissDySPAN2010}, there has been significant
current research efforts in the context of an European project
SENDORA \cite{SENDORA2008}, which aims at developing techniques
based on sensor networks for supporting coexistence of licensed and
unlicensed wireless users in a same area. The users are equipped
with software defined radios and can tune to transmit in a wide
range of frequencies as instructed by the operator, but do not
necessarily have the cognitive sensing capacity. Since the secondary
users do not worry about sensing, they can spend most of their time
and energy on actual data transmissions. Such a network structure
puts most of the implementation complexity at the operator side and
reduces the user equipment complexity, and thus might be easier to
implement in practice than a ``full'' cognitive network. }

The key notations of this paper are listed in Table
\ref{tab:notation} with some explanations as follows.
\begin{itemize}
\item \emph{Investment decisions $B_{s}$ and $B_{l}$}: the operator's sensing and leasing bandwidths, respectively.

\item \emph{Sensing realization factor $\alpha$}: when the operator senses a total bandwidth of $B_s$, only a
proportion of $\alpha\in[0,1]$ is unused and can be used by the
operator. $\alpha$ is a random variable and depends on the primary
users' activities. With perfect sensing results, users can use
bandwidth up to  $B_s\alpha$ without generating interferences to the
primary users.

\item \emph{Cost parameters $C_{s}$ and $C_{l}$}: the operator's fixed sensing and  leasing costs per unit bandwidth, respectively.
Sensing cost $C_{s}$ depends on the operator's sensing technologies.
When the operator senses spectrum, it needs to spend time and energy
on channel sampling and signal processing (\cite{YCLiang}). Sensing
over different channels often needs to be done sequentially due to
the potentially large number of channels open to opportunistic
spectrum access and the limited power/hardware capacity of cognitive
radios (\cite{Shu}). The larger sensing bandwidth and the more
channels, the longer time and  higher energy it requires
(\cite{IEEEhowto:Chen}). For simplicity, we assume that total
sensing cost is linear in the sensing bandwidth $B_s$. Leasing cost
$C_{l}$ is determined through the negotiation between the operator
and the spectrum owner and is assumed to be larger than $C_{s}$.

\item \emph{Pricing decision $\pi$}: the operator's choice of  price per unit bandwidth to the users. \end{itemize}
%

%



\begin{figure}[tt]
\centering
\includegraphics[width=0.22\textwidth,angle=90]{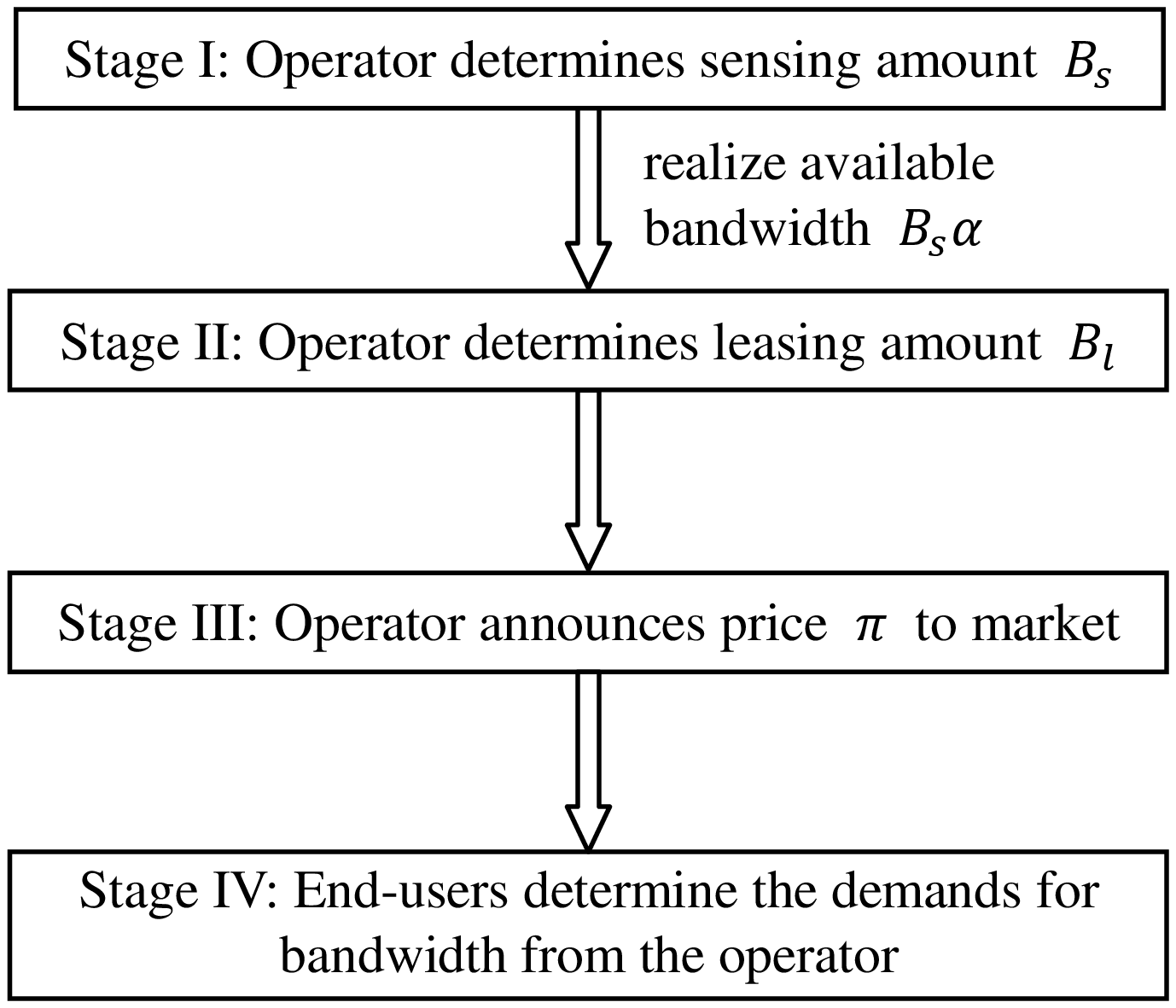}
\caption{A {Stackelberg Game}} \label{fig:fourstage}
\end{figure}

\subsection{A {Stackelberg Game}}

We consider a {Stackelberg Game} between the operator and the users
as shown in Fig.~\ref{fig:fourstage}. The operator is the
Stackelberg leader: it first decides the sensing amount $B_s$ in
Stage I, then decides the leasing amount $B_l$ in Stage II (based on
the sensing result $B_{s}\alpha$), and then announces the price
$\pi$ to the users in Stage III (based on the total supply
$B_s\alpha+B_l$). Finally, the users choose their bandwidth demands
to maximize their individual payoffs in Stage IV.

{We note that ``sensing followed by leasing'' is optimal for the
operator to maximize profit. Since sensing is cheaper than leasing,
the operator should lease only if sensing does not provide enough
resource. If the operator determines sensing and leasing
simultaneously, then it is likely  to ``over-lease'' (compared with
``sensing followed by leasing'') to avoid having too little resource
when $\alpha$ is small. ``Leasing before sensing'' can not improve
the operator's profit either due to the same reason. }

\section{Backward Induction of the Four-stage Game}
\label{sec:BackwardInduction}

The {Stackelberg game} falls into the class of dynamic game, and the
common solution concept is the Subgame Perfect Equilibrium (SPE, or
simply as \emph{equilibrium} in this paper). A general technique for
determining the SPE is the backward induction
(\cite{microeconomic_book}). We will start with Stage IV and analyze
the users' behaviors given the operator's investment and pricing
decisions. Then we will look at Stage III and analyze how the
operator makes the pricing decision given investment decisions and
the possible reactions of the users in Stage IV. Finally we proceed
to derive the operator's optimal leasing decision in Stage II and
then the optimal sensing decision in Stage I. The backward induction
captures the sequential dependence of the decisions in four stages.

\subsection{Spectrum Allocation in Stage IV}\label{subsec:stageIV}
In Stage IV, end-users determine their bandwidth demands given the
unit price $\pi$ announced by the operator in stage III.

\com{
\begin{figure}[htbp]
\centering
\includegraphics[width=0.27\textwidth]{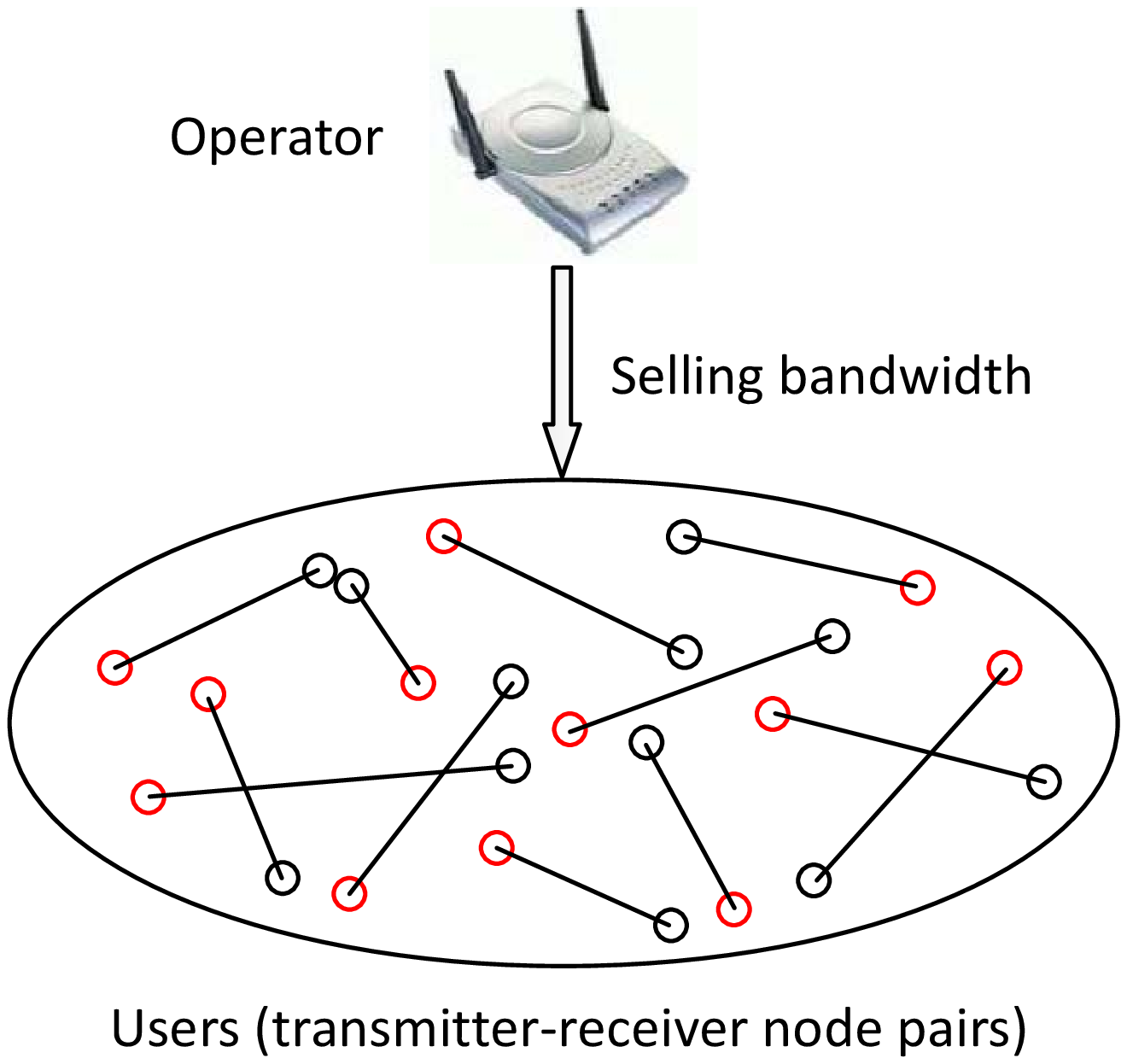}
\caption{Network scenario with secondary users forming an ad hoc
network} \label{fig:user1}
\end{figure}}

\com{
\begin{figure}[htbp]
\centering
\includegraphics[width=0.35\textwidth]{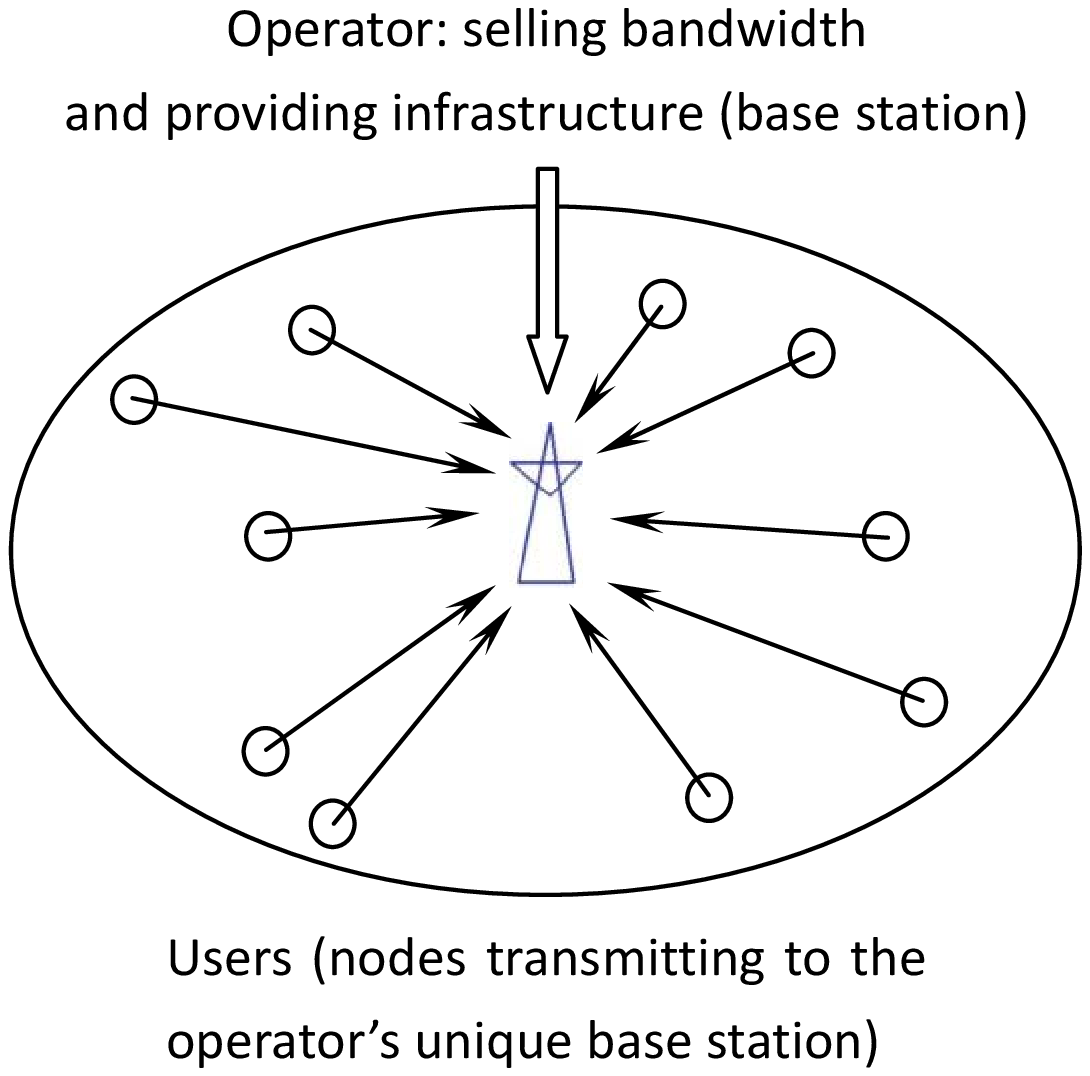}
\caption{Network scenario with secondary users transmitting to the
base station of the cognitive virtual mobile network operator}
\label{fig:user2}
\end{figure}
}

Each user can represent a transmitter-receiver node pair in an ad
hoc network\com{(e.g., Fig.~\ref{fig:user1})}, or a node that
transmits to the operator's base station in an uplink scenario\com{(e.g., Fig.~\ref{fig:user2})}. We assume that users access the
spectrum provided by the operator through FDM (Frequency-division
multiplexing) or OFDM (Orthogonal frequency-division multiplexing)
to avoid mutual interferences. User $i$'s achievable rate (in nats)
is\footnote{{We assume that the operator only provides bandwidth
without restricting the application types. This assumption has been
commonly used in dynamic spectrum sharing literature, e.g.,
\cite{IEEEhowto:Niyato,Simeone,NiyatoTMC2009,IEEEhowto:Jia}.}}:
\begin{equation}\label{eq:rate}
r_i(w_i)= w_i
\ln\left(1+\frac{P_i^{\max}h_i}{n_0w_i}\right),\end{equation}
where $w_i$ is the allocated bandwidth from the operator,
$P_i^{\max}$ is user $i$'s maximum transmission power, $n_0$ is the
noise power per unit bandwidth, $h_i$ is user $i$'s channel gain
(between user $i$'s own transmitter and receiver in an ad hoc
network, or between user $i$'s transmitter to the operator's base
station in an uplink scenario). To obtain rate in (\ref{eq:rate}),
user $i$ spreads its maximum transmission power $P_k^{\max}$ across
the entire allocated bandwidth $w_{i}$. To simplify the notation, we
let $g_i = P_i^{\max}h_i/n_0$, thus $g_i/w_i$ is the user $i$'s
signal-to-noise ratio (SNR). Here we focus on best-effort users who
are interested in maximizing their data rates. Each user only knows
its local information (i.e., $P_i^{\max}$, $h_i$, and $n_0$) and
does not know anything about other users.

{From a user's point of view, it does not matter whether the
bandwidth has been obtained by the operator through spectrum sensing
or dynamic leasing. Each unit of allocated bandwidth is perfectly
reliable for the user.}

To obtain closed-form solutions, we  \emph{first focus on the high
SNR regime where $\mathtt{SNR}\gg1$}. This is motivated by the fact
that users often have limited choices of modulation and coding
schemes, and thus may not be able to decode a transmission if the
SNR is below a threshold.
In the high SNR regime, the rate in (\ref{eq:rate}) can be
approximated as
\begin{equation}\label{eq:highSNR}
r_i(w_i)= w_i \ln\left(\frac{g_i}{w_i}\right).
\end{equation}
Although the analytical solutions in Section
\ref{sec:BackwardInduction} are derived based on (\ref{eq:highSNR}),
we emphasize that \emph{all the major engineering insights remain
true in the general SNR regime.} A formal proof is in Section
\ref{sec:Equilibrium}.

A user $i$'s \emph{payoff} is a function of the allocated bandwidth
$w_i$ and the price $\pi$,
\begin{equation}\label{eq:utility}
u_i(\pi,w_i) = w_i \ln\left(\frac{g_i}{w_i}\right)-\pi w_i,
\end{equation}
i.e., the difference between the data rate and the linear payment
($\pi w_i$). Payoff $u_i(\pi,w_i)$ is concave in $w_i$, and the
unique bandwidth \emph{demand} that maximizes the payoff is
\begin{equation}\label{eq:optbandwidth}
w_{i}^*(\pi)=\arg\max_{w_{i}\geq
0}u_{i}(\pi,w_{i})=g_{i}e^{-(1+\pi)},
\end{equation}
which is always positive, linear in $g_i$, and decreasing in price
$\pi$.  Since $g_i$ is linear in channel gain $h_i$ and transmission
power $P_i^{\max}$, then a user with a better channel condition or a
larger transmission power has a larger demand.


Equation (\ref{eq:optbandwidth}) shows that each user $i$ achieves
the same SNR:
$$\mathtt{SNR}_{i}=\frac{g_i}{w_i^*(\pi)}=e^{(1+\pi)}.$$
but a different payoff that is linear in $g_i$,
$$u_{i}(\pi,w_i^\ast(\pi))=g_i e^{-(1+\pi)}. $$
We denote users' aggregate wireless characteristics as
$G=\sum_{i\in\mathcal{I}}g_i$. The users' total demand is
\begin{equation}\label{eq:demand}
\sum_{i\in \mathcal{I}}w_{i}^{*}(\pi)=Ge^{-(1+\pi)}.
\end{equation}

Next, we will consider how the operator makes the investment
(sensing and leasing) and pricing decisions in Stages I-III based on
the total demand in eq. (\ref{eq:demand})\footnote{We assume that
the operator knows the value of $G$ through proper feedback
mechanism from the users.}. In particular, we will show that the
operator will always choose a price in Stage III such that the total
demand (as a function of price) does not exceed the total supply.

\begin{table*}[tt]
\centering \caption{Optimal Pricing Decision and Profit in Stage
III}
\begin{tabular}{|p{1.6in}|l|p{2.5in}|}
\hline Total Bandwidth Obtained in Stages I and II & Optimal Price
$\pi^*\left(B_{s},\alpha,B_{l}\right)$ & Optimal Profit
$R_{III}(B_s,\alpha,B_l)$\\\hline Excessive Supply Regime:
\;\;\;\;\;\;\;\;\;\;\;$B_{l}+B_{s}\alpha\geq Ge^{-2}$ & $\pi^{ES}=1$
&
$R_{III}^{ES}(B_{s},\alpha,B_l)=Ge^{-2}-B_{s}C_{s}-B_{l}C_{l}$\\\hline
Conservative Supply Regime: $  B_{l}+B_{s}\alpha < Ge^{-2}$ &
$\pi^{CS}=\ln\left(\frac{G}{B_{l}+B_{s}\alpha}\right)-1$ &
$R_{III}^{CS}(B_{s},\alpha,B_l)=(B_l+B_{s}\alpha)\ln\left(\frac{G}{B_{l}+B_{s}\alpha}\right)-B_s(\alpha+C_s)-B_l(1+C_l)$
\\\hline
\end{tabular}
\label{tab:pricing}
\end{table*}

\subsection{Optimal Pricing Strategy in Stage III}
\label{Sect:stageIII} In Stage III, the operator determines the
optimal pricing considering users' total demand (\ref{eq:demand}),
given the bandwidth supply $B_{s}\alpha + B_{l}$ obtained
in Stage II. The operator profit is%
\begin{multline}\label{eq:profit}
R(B_s,\alpha,B_l,\pi) =\min\left(\pi \sum_{i\in
\mathcal{I}}w_{i}^{*}(\pi), \pi\left(B_{l}+B_{s}\alpha\right)\right)
\\-\left(B_{s}C_{s}+B_{l}C_{l}\right),
\end{multline}
which is the difference between the revenue and total cost. The
$\min$ operation denotes the fact that the operator can only satisfy
the demand up to its available supply. The objective of Stage III is
to find the optimal price $\pi^\ast\left(B_{s},\alpha,B_{l}\right)$
that maximizes the profit, that is,
\begin{equation}\label{eq:profitStateIII}
R_{III}(B_s,\alpha,B_l)=\max_{\pi\geq 0}R(B_s,\alpha,B_l,\pi).
\end{equation}
The subscript ``III'' denotes the best profit in Stage III.

\begin{figure}[tt]
\centering
\includegraphics[width=0.3\textwidth]{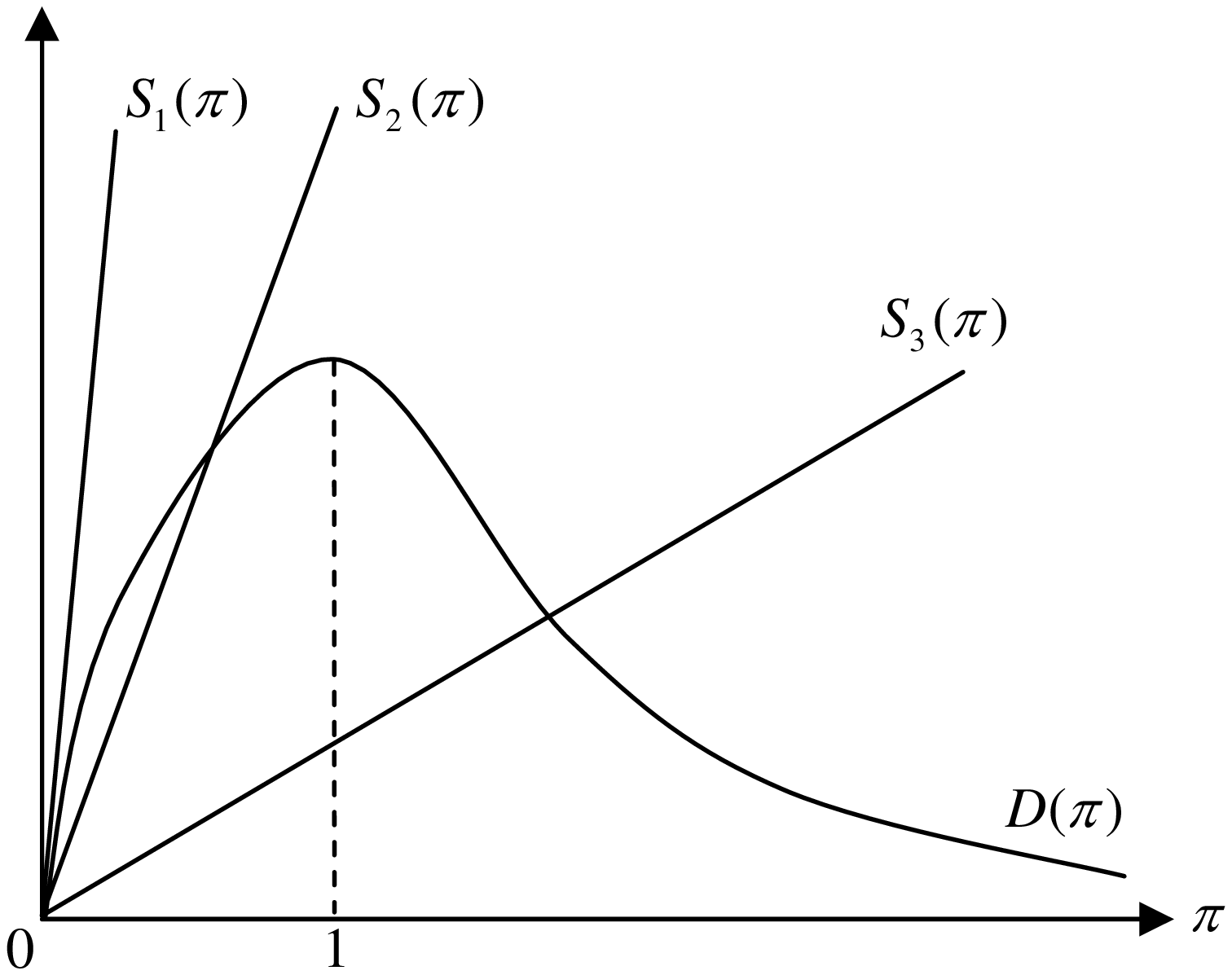}
\caption{Different intersection cases of $D(\pi)$ and $S(\pi)$}
\label{fig:intersection}
\end{figure}

Since the bandwidths $B_s$ and $B_l$ are given in this stage, the
total cost $B_sC_s+B_lC_l$ is already fixed. The only optimization
is to choose the optimal price $\pi$ to maximize the revenue, i.e.,
\begin{equation}\label{eq:maxmin}
\max_{\pi\geq 0} \min\left(\pi \sum_{i\in
\mathcal{I}}w_{i}^{*}(\pi),
\pi\left(B_{l}+B_{s}\alpha\right)\right).
\end{equation}

The solution of problem (\ref{eq:maxmin}) depends on the bandwidth
investment in Stages I and II. Let us define $D(\pi)=\pi \sum_{i \in
\mathcal{I}}w_i^*(\pi)$ and $S(\pi)=\pi(B_l+B_s\alpha)$. Figure
\ref{fig:intersection} shows three possible relationships between
these two terms, where $S_{j}(\pi)$ (for $j=1,2,3$) represents each
of the three possible choices of $S(\pi)$ depending on the bandwidth
$B_l+B_s\alpha$:
\begin{itemize}
\item $S_1(\pi)$ (excessive supply): No intersection with $D(\pi)$;
\item $S_2(\pi)$ (excessive supply): intersect once with $D(\pi)$ where $D(\pi)$ has a non-negative slope;
\item $S_3(\pi)$ (conservative supply): intersect once with $D(\pi)$ where $D(\pi)$ has a negative slope.
\end{itemize}


%
In the \emph{excessive supply} regime, $\max_{\pi\geq
0}\min\left(S(\pi),D(\pi)\right)=\max_{\pi\geq 0} D(\pi)$, i.e., the
max-min solution occurs at the maximum value of $D(\pi)$ with
$\pi^{\ast}=1$. In this regime, the total supply is larger than the
total demand at the best price choice. In the \emph{conservative
supply} regime, the max-min solution occurs at the unique
intersection point of $D(\pi)$ and $S(\pi)$. The above observations
lead to the following result.

\begin{theorem}\label{Thm:pricing}
The optimal pricing decision and the corresponding optimal profit at
Stage III can be characterized by Table \ref{tab:pricing}.
\end{theorem}

The proof of Theorem \ref{Thm:pricing} is given in Appendix
\ref{app:thm1}. Note that in the excessive supply regime, some
bandwidth is left unsold (i.e., $S(\pi^{\ast})>D(\pi^{\ast})$). This
is because the acquired bandwidth is too large, and selling all the
bandwidth will lead to a very low price that decreases the revenue
(the product of price and sold bandwidth). The profit can be
apparently improved if the operator acquires less bandwidth in
Stages I and II. Later analysis in Stages II and I will show that
the equilibrium of the game must lie in the conservative supply
regime if the sensing cost is non-negligible.

\subsection{Optimal Leasing Strategy in Stage II}
\label{Sect:stageII}

In Stage II, the operator decides the optimal leasing amount $B_l$
given the sensing result $B_s \alpha$:
\begin{equation}\label{eq:R_II}R_{II}(B_s,\alpha)=\max_{B_l\geq
0}R_{III}(B_s,\alpha,B_l).\end{equation}
We decompose problem (\ref{eq:R_II}) into two subproblems based on
the two supply regimes in Table \ref{tab:pricing},
\begin{enumerate}
\item Choose $B_{l}$ to reach the excessive supply regime in Stage III:
\begin{equation}\label{eq:leasingSub2}
R_{II}^{ES}(B_s,\alpha)=\max_{B_l\geq
\max\left\{Ge^{-2}-B_s\alpha,0\right\}}R_{III}^{ES}(B_s,\alpha,B_l).
\end{equation}
\item Choose $B_{l}$ to reach the conservative supply regime in Stage III:
\begin{equation}\label{eq:leasingSub1}
R_{II}^{CS}(B_s,\alpha)=\max_{0\leq B_l\leq
Ge^{-2}-B_s\alpha}R_{III}^{CS}(B_s,\alpha,B_l),
\end{equation}
\end{enumerate}

\begin{table*}[t]
\centering \caption{Optimal Leasing Decision and Profit in Stage II}
\begin{tabular}{|l|l|l|} \hline Given Sensing Result $B_s\alpha$ After Stage
I &Optimal Leasing Amount $B_l^*$ & Optimal Profit
$R_{II}(B_s,\alpha)$
\\\hline (CS1)  $ B_s\alpha\leq {G}{e^{-(2+C_l)}}$ & $B_l^{CS1}={G}{e^{-(2+C_l)}}-B_s\alpha$ &
$R_{II}^{CS1}(B_s,\alpha)  =G{e^{-(2+C_l)}}+B_{s}(\alpha
C_{l}-C_{s})$
\\\hline (CS2) $ B_s\alpha\in\left({G}{e^{-(2+C_l)}}, Ge^{-2}\right]$ &$B_l^{CS2}=0$ &
$R_{II}^{CS2}(B_s,\alpha)  =B_{s}\alpha
\ln\left(\frac{G}{B_{s}\alpha}\right)-B_{s}(\alpha +C_{s})$
\\\hline (ES3) $B_s\alpha>Ge^{-2}$ &$B_l^{ES3}=0$ &$R_{II}^{ES3}(B_s,\alpha)  =Ge^{-2}-B_{s}C_{s}$\\\hline
\end{tabular}
\label{tab:leasing}
\end{table*}

To solve subproblems (\ref{eq:leasingSub2}) and
(\ref{eq:leasingSub1}), we need to consider the bandwidth obtained
from sensing.

\begin{itemize}

\item \emph{Excessive Supply} ($B_s\alpha>Ge^{-2}$): in this case,
the feasible sets of both subproblems (\ref{eq:leasingSub2}) and
(\ref{eq:leasingSub1}) are empty. In fact, the bandwidth supply is
already in the excessive supply regime as defined in Table II, and
it is optimal not to lease in Stage II.

\item \emph{Conservative Supply} ($B_s\alpha\leq Ge^{-2}$): first,
we can show that the unique optimal solution of subproblem
(\ref{eq:leasingSub2}) is $B_l^\ast = Ge^{-2}-B_s\alpha$. This means
that the optimal objective value of subproblem
(\ref{eq:leasingSub2}) is no larger than that of subproblem
(\ref{eq:leasingSub1}), and thus it is enough to consider subproblem
(\ref{eq:leasingSub1}) in the conservative supply regime only.

\end{itemize}

Base on the above observations and some further analysis, we can
show the following:
\begin{theorem}\label{thm:leasing}
In Stage II, the optimal leasing decision and the corresponding
optimal profit are summarized in Table \ref{tab:leasing}.
\end{theorem}


The proof of Theorem \ref{thm:leasing} is given in Appendix
\ref{app:thm2}. Table \ref{tab:leasing} contains three cases based
on the value of $B_{s}\alpha$:  (CS1), (CS2), and (ES3). The first
two cases involve solving the subproblem (\ref{eq:leasingSub1}) in
the conservative supply regime, and the last one corresponds to the
excessive supply regime. Although the decisions in cases (CS2) and
(ES3) are the same (i.e., zero leasing amount), we still treat them
separately since the profit expressions are different.

It is clear that we have an optimal \emph{threshold} leasing policy
here: the operator wants to achieve a total bandwidth equal to
$Ge^{-(2+C_{l})}$ whenever possible. When the bandwidth obtained
through sensing is not enough, the operator will lease additional
bandwidth to reach the threshold; otherwise the operator will not
lease.

\subsection{Optimal Sensing Strategy in Stage I}\label{subsec:stageI}
In Stage I, the operator will decide the optimal sensing amount to
maximize its expected profit by taking the uncertainty of the
sensing realization factor $\alpha$ into account. The operator needs
to solve the following problem $$R_{I} =  \max_{B_{s \geq 0}}
R_{II}\left(B_{s}\right),$$ where $R_{II}\left(B_{s}\right)$ is
obtained by taking the expectation of $\alpha$ over the profit
functions in Stage II (i.e., $R_{II}^{CS1}(B_{s},\alpha)$,
$R_{II}^{CS2}(B_{s},\alpha)$, and $R_{II}^{ES3}(B_{s},\alpha)$ in
Table \ref{tab:leasing}).


\com{
\begin{figure}[tt]
\centering
\includegraphics[width=0.3\textwidth]{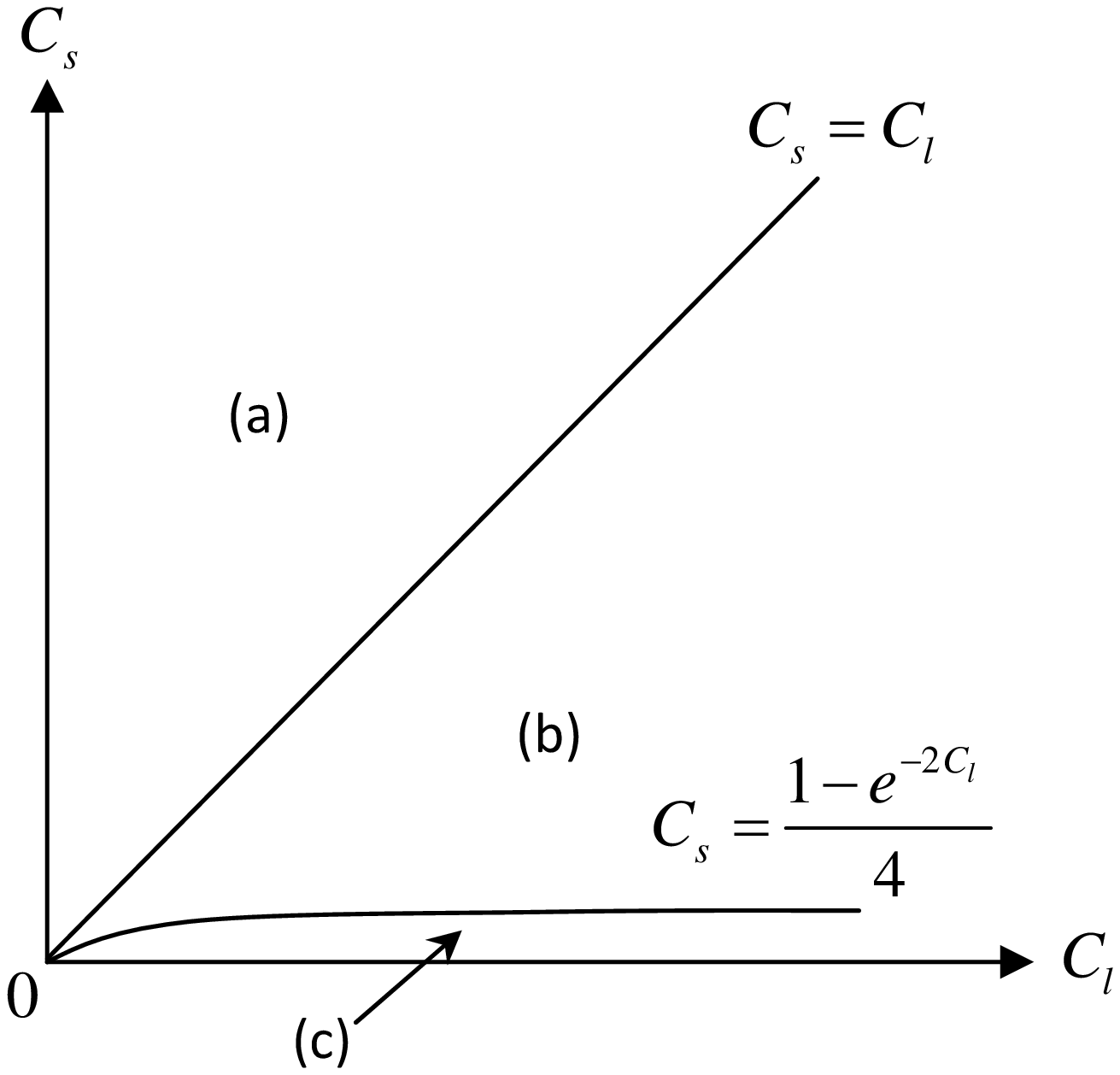}
\caption{The feasible region of the sensing and leasing costs
$(C_{s},C_{l})$, divided into three sub-regions depending on the
relationship between the two costs. } \label{fig:lowerboundCs}
\end{figure}
}

To obtain closed-form solutions, we assume that the sensing
realization factor $\alpha$ follows a uniform distribution in
$[0,1]$. In Section \ref{subsec:robustness}, we prove that \emph{the
major engineering insights also hold under any general
distribution.}
%

To avoid the trivial case where sensing is so cheap that it is
optimal to sense a huge amount of bandwidth, we further assume that
the sensing cost is non-negligible and is lower bounded by $C_s \geq
(1-e^{-2C_l})/4$.
%
\com{As shown in Fig.~\ref{fig:lowerboundCs}, we will ignore region
(c) which is very small compared with the entire feasible region
((a), (b), and (c)).}

To derive function $R_{II}\left(B_{s}\right)$, we will consider the
following three intervals:
%
%

\begin{table*}
\centering \caption{Choice of Optimal Sensing Amount in Stage I}
\begin{tabular}{|p{3in}|p{1.7in}|p{1.5in}|}
\hline \multicolumn{1}{|l|}{} &\multicolumn{1}{|l|}{Optimal Sensing
Decision $B_s^*$} &\multicolumn{1}{|c|}{Expected Profit $R_{I}$}
\\\hline High Sensing Cost Regime: $C_s\geq C_l/2$
&$B_s^*=0$ &$R_I^H={G}{e^{-(2+C_l)}}$
\\\hline Low Sensing Cost Regime: $ C_s \in\left[(1-e^{-2C_l})/4,
C_l/2\right]$ &$B_s^*=B_{s}^{L*}$, solution to eq.\
(\ref{eq:LowSensingAmount}) &$R_{I}^{L}$ in eq.
(\ref{eq:ProfitLowSensing})
\\\hline
\end{tabular}
\tabcolsep 5mm \label{tab:sensing}
\end{table*}

\begin{enumerate}
\item Case I: $B_s\in[0, G{e^{-(2+C_l)}}]$. In this case, we always have
$B_s\alpha\leq G{e^{-(2+C_l)}}$ for any value $\alpha\in[0,1]$,
which corresponds to case (CS1) in Table \ref{tab:leasing}. The
expected profit is
\begin{align}
R_{II}^{1}(B_{s})=&E_{\alpha\in[0,1]}\left[R_{II}^{CS1}(B_{s},\alpha)\right]\nonumber\\=&G{e^{-(2+C_l)}}+B_{s}\left(\frac{C_l}{2}-C_s\right),\nonumber
\end{align}
which is a linear function of $B_s$. If $C_s> C_l/2$,
$R_{II}^1(B_s)$ is linearly decreasing in $B_s$; if $C_s <C_l/2$,
$R_{II}^1(B_s)$ is linearly increasing in $B_s$.
\item Case II: $
B_s\in\left(G{e^{-(2+C_l)}}, G{e^{-2}}\right]$. Depending on the
value of $\alpha$, $B_{s}\alpha$ can be in either case (CS1) or case
(CS2) in Table \ref{tab:leasing}.
%
%
The expected profit is
\begin{align}
R_{II}^{2}(B_{s})=&
E_{\alpha\in\left[0,\frac{G{e^{-(2+C_l)}}}{B_s}\right]}\left[R_{II}^{CS1}(B_{s},\alpha)
\right]\nonumber\\
&+E_{\alpha\in\left[\frac{G{e^{-(2+C_l)}}}{B_s},1\right]}\left[R_{II}^{CS2}(B_{s},\alpha)
\right]\nonumber
\\=&\frac{B_s}{2} \ln \left(\frac{G}{B_s}\right)-\frac{B_s}{4}+\frac{B_s}{4}\left(\frac{G
e^{-(2+C_l)}}{B_s}\right)^2-B_{s}C_s.\nonumber
\end{align} $R_{II}^2(B_s)$ is a strictly concave function of $B_s$ since its second-order derivative  $$\frac{\partial^2 R_{II}^2(B_s)}{\partial B_s^2}=\frac{1}{2B_s}\left[
\left(\frac{Ge^{-(2+C_l)}}{B_s}\right)^2-1\right]<0$$ as
$B_s>Ge^{-(2+C_l)}$ in this case.

\begin{figure}[tt]
\centering
\includegraphics[width=0.35\textwidth]{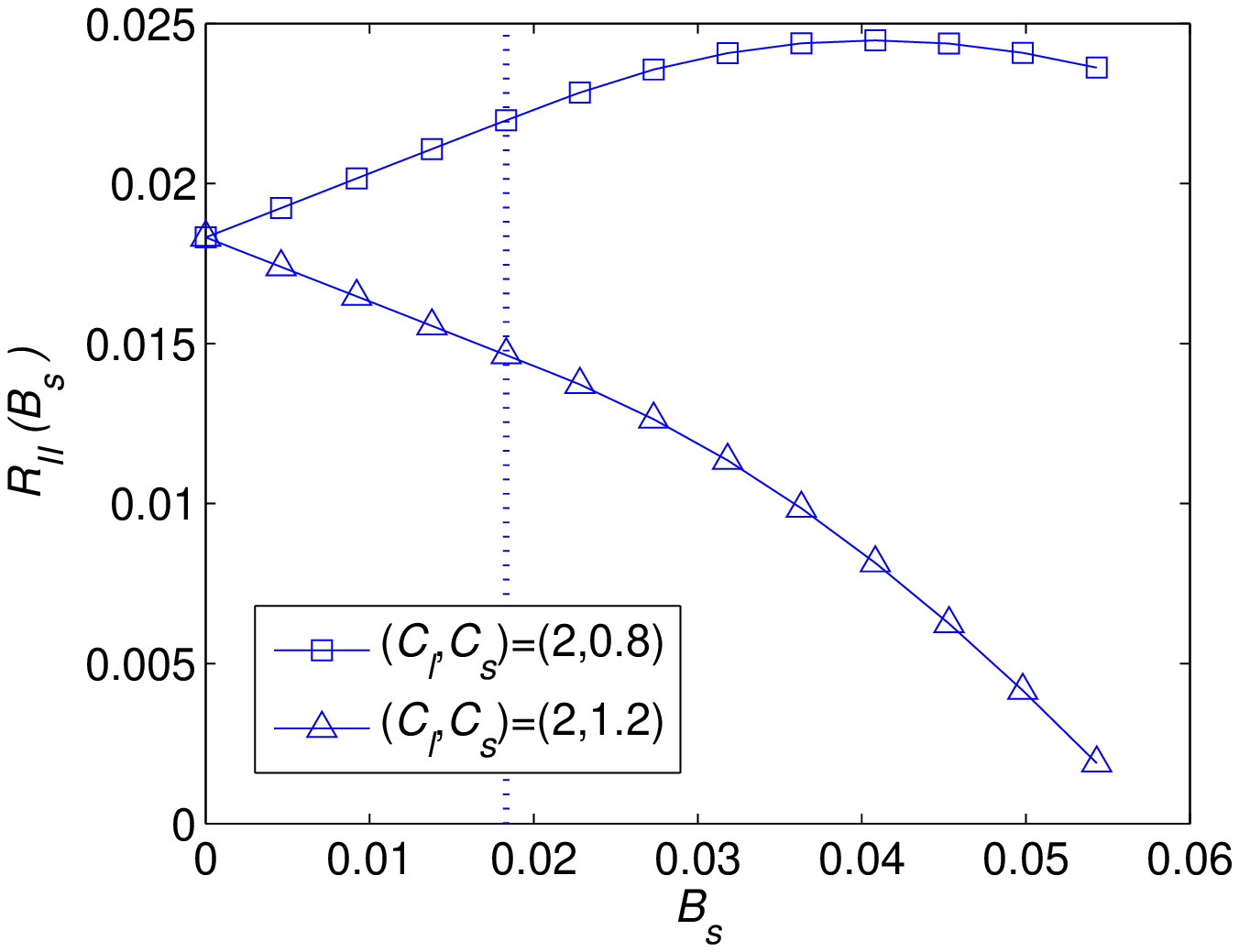}
\caption{Expected profit in Stage II under different sensing and
leasing costs} \label{fig:sensing}
\end{figure}

\item Case III: $B_s\in\left(G{e^{-2}},\infty\right)$. Depending on the value of $\alpha$, $B_{s}\alpha$
can be any of the three cases in Table \ref{tab:leasing}.
%
%
The expected profit is
\begin{align}
R_{II}^{3}(B_s)=&
E_{\alpha\in\left[0,\frac{G{e^{-(2+C_l)}}}{B_s}\right]}\left[R_{II}^{CS1}(B_{s},\alpha)
\right]\nonumber\\&+E_{\alpha\in\left[\frac{G{e^{-(2+C_l)}}}{B_s},\frac{G{e^{-2}}}{B_s}\right]}\left[R_{II}^{CS2}(B_{s},\alpha)
\right]\nonumber\\&+E_{\alpha\in\left[\frac{G{e^{-2}}}{B_s},1\right]}\left[R_{II}^{ES3}(B_{s},\alpha)
\right]\nonumber\\=& \left(\frac{G}{e^2}\right)^2
\frac{e^{-2C_l}-1}{4B_s}-B_sC_s+\frac{G}{e^2}.\nonumber\end{align}
Because its first-order derivative
$$\frac{\partial R_{II}^3(B_s)}{\partial B_s}=\left(\frac{Ge^{-2}}{B_s}\right)^2\frac{1-e^{-2C_l}}{4}-C_s<0,$$
as $B_s>Ge^{-2}$ in this case, $R_{II}^3(B_s)$ is decreasing in
$B_s$ and achieves its maximum at $B_s=Ge^{-2}$.
\end{enumerate}

To summarize, the operator needs to maximize
\begin{equation}\label{eq:sensingdomain} R_{II}(B_s)=
\begin{cases}
R_{II}^{1}(B_{s}), & \text{if } 0\leq B_{s}\leq Ge^{-(2+C_l)}; \\
R_{II}^{2}(B_{s}), & \text{if } Ge^{-(2+C_l)}<B_{s}\leq
Ge^{-2};\\
R_{II}^{3}(B_{s}), & \text{if } B_{s}> Ge^{-2}. \\
\end{cases}
\end{equation}

We can verify that Case II always achieves a higher optimal profit
than Case III. This means that the optimal sensing will only lead to
either case (CS1) or case (CS2) in Stage II, which corresponds to
the conservative supply regime in Stage III. This confirms our
previous intuition that equilibrium is always in the conservative
supply regime under a non-negligible sensing cost, since some
resource is wasted in the excessive supply regime (see discussions
in Section \ref{Sect:stageIII}).

Table \ref{tab:sensing} shows that the sensing decision is made in
the following two cost regimes:
\begin{itemize}
\item \emph{High sensing cost regime} ($C_{s}>C_{l}/2$): it is optimal not to sense.
Intuitively, the coefficient $1/2$ is due to the uniform
distribution assumption of $\alpha$, i.e., on average obtaining one
unit of available bandwidth through sensing costs $2C_{s}$.
\item \emph{Low sensing cost regime} ($C_s\in\left[\frac{1-e^{-2C_l}}{4},\frac{C_l}{2}\right]$):
the optimal sensing amount $B_{s}^{L*}$ is the unique solution to
the following equation:
\begin{equation}\label{eq:LowSensingAmount}
\frac{\partial R_{II}^2(B_s)}{\partial
B_s}=\frac{1}{2}\ln\left(\frac{1}{B_s/G}\right)-\frac{3}{4}-C_s-\left(\frac{e^{-(2+C_l)}}{2B_s/G
}\right)^2=0.
\end{equation}%
The uniqueness of the solution is due to the strict concavity of
$R_{II}^2(B_s)$ over $B_s$. We can further show that $B_{s}^{L*}$
lies in the interval of $\left[{G}{e^{-(2+C_l)}},Ge^{-2}\right]$ and
is linear in $G$.
Finally, the operator's optimal expected profit is
\begin{equation}\label{eq:ProfitLowSensing}
R_{I}^{L}=\frac{B_{s}^{L*}}{2} \ln
\left(\frac{G}{B_{s}^{L*}}\right)-\frac{B_{s}^{L*}}{4}
+\frac{1}{4B_{s}^{L*}}\left(\frac{G}{
e^{2+C_l}}\right)^2-B_{s}^{L*}C_s.
\end{equation}
\end{itemize}

Based on these observations, we can show the following:

\begin{theorem}
In Stage I, the optimal sensing decision and the corresponding
optimal profit are summarized in Table \ref{tab:sensing}. The
optimal sensing amount $B_l^*$ is linear in $G$.
\end{theorem}

\begin{table*}
\centering \caption{The Operator's and Users' Equilibrium Behaviors}
\begin{tabular}{|c|c|c|c|} \hline \multicolumn{1}{|c|}{Sensing Cost Regimes}
&\multicolumn{1}{|c|}{High Sensing Cost: $C_s\geq \frac{C_l}{2}$}
&\multicolumn{2}{|c|}{Low Sensing Cost: $\frac{1-e^{-2C_l}}{4}\leq
C_s\leq \frac{C_l}{2}$}
\\\hline \multicolumn{1}{|c|}{Optimal Sensing
Amount $B_s^*$} &\multicolumn{1}{|c|}{$0$}
&\multicolumn{2}{|c|}{$B_{s}^{L*}\in\left[Ge^{-(2+C_l)},Ge^{-2}\right]$,
solution to eq. (\ref{eq:LowSensingAmount})} \\\hline {Sensing
Realization Factor $\alpha$} &$0\leq\alpha\leq1$ &$0\leq\alpha\leq
Ge^{-(2+C_l)}/B_{s}^{L*}$ &$\alpha> Ge^{-(2+C_l)}/{B_{s}^{L*}}$
\\\hline Optimal Leasing Amount $B_l^*$&$G{e^{-(2+C_l)}}$ &{$G{e^{-(2+C_l)}}-B_{s}^{L*}\alpha$} &{$0$}
\\\hline Optimal Pricing $\pi^{*}$ &$1+C_l$ &$1+C_l$
&$\ln\left(\frac{G}{B_{s}^{L*}\alpha}\right)-1$\\\hline {Expected
Profit $R_I$} &\multicolumn{1}{|c|}{$R_{I}^{H}=G{e^{-(2+C_l)}}$}
 &\multicolumn{1}{|c|}{$R_I^{L}$ in
eq. (\ref{eq:ProfitLowSensing})} &\multicolumn{1}{|c|}{$R_I^{L}$ in
eq. (\ref{eq:ProfitLowSensing})}
\\\hline
User \emph{i}'s SNR &$e^{(2+C_l)}$ &$e^{(2+C_l)}$
&$\frac{G}{B_s^{L*}\alpha}$
\\\hline
User \emph{i}'s Payoff &$g_i e^{-(2+C_l)}$ &$g_i e^{-(2+C_l)}$ &$g_i
(B_s^{L*}\alpha/G)$
\\\hline
\end{tabular}
\tabcolsep 5mm \label{tab:equilibrium}
\end{table*}

Figure~\ref{fig:sensing} shows two possible cases for the function
$R_{II}(B_s)$. The vertical dashed line represents
$B_s=e^{-(2+C_l)}$. For illustration purpose, we assume $G=1$,
$C_{l}=2$, and $C_{s}=\{0.8, 1.2\}$. When the sensing cost is large
(i.e., $C_{s}=1.2>C_{l}/2$), $R_{II}(B_s)$ achieves its optimum at
$B_{s}=0$ and thus it is optimal not to sense. When the sensing cost
is small (i.e., $C_{s}=0.8<C_{l}/2$), $R_{II}(B_s)$ achieves its
optimum at $B_{s}>e^{-(2+C_l)}$ and it is optimal to sense a
positive amount of spectrum.

\section{Equilibrium Summary and Numerical Results}
\label{sec:Equilibrium}

\begin{figure*}
   \begin{minipage}[t]{0.32\linewidth}
      \centering
      \includegraphics[width=1.1\textwidth]{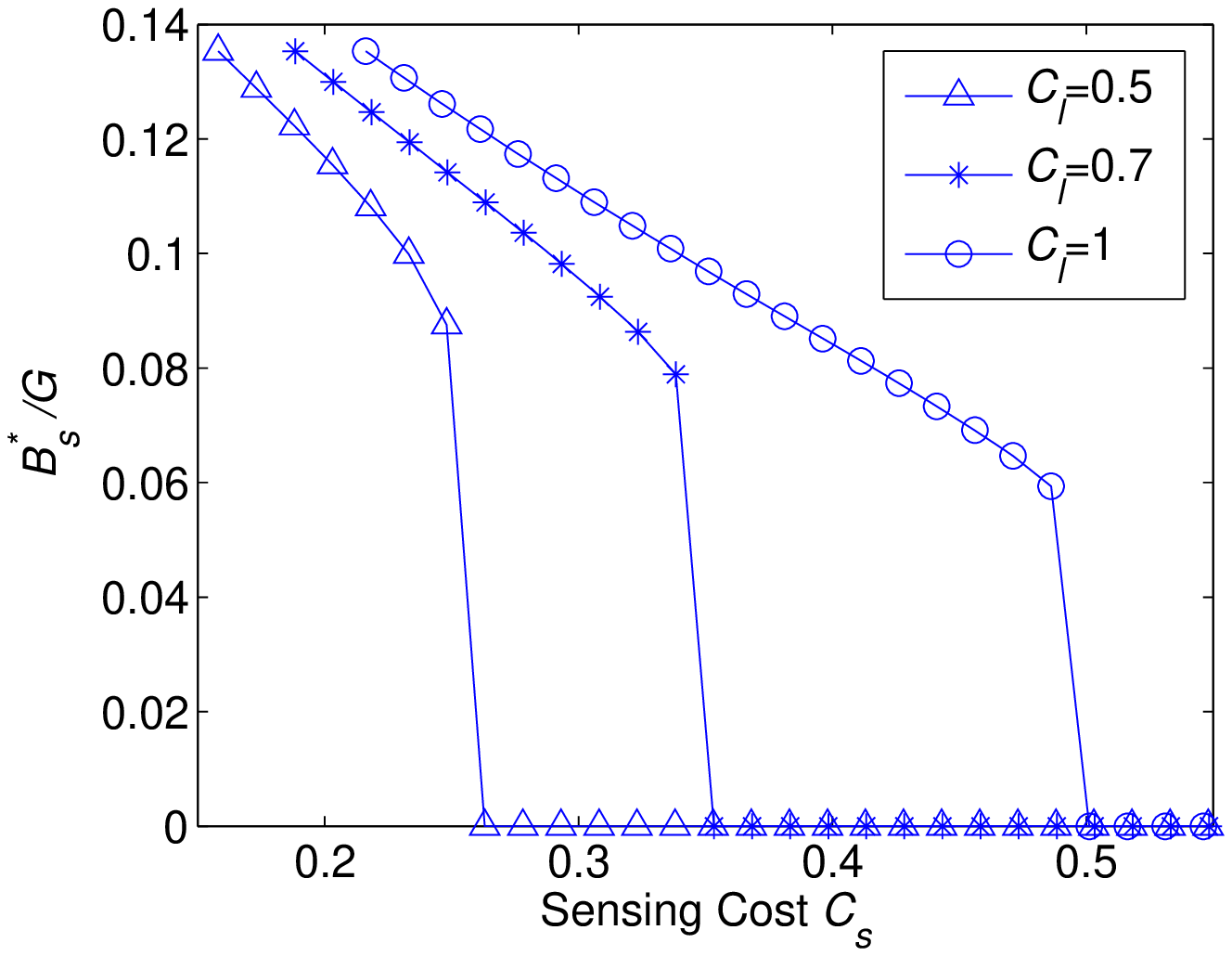}
      \caption{Optimal sensing amount $B_s^*$ as a function of \com{sensing
cost }$C_{s}$ and \com{leasing cost }$C_{l}$.} \label{fig_sim1}
   \end{minipage}
   \begin{minipage}[t]{0.32\linewidth}
      \centering
      \includegraphics[width=1.1\textwidth]{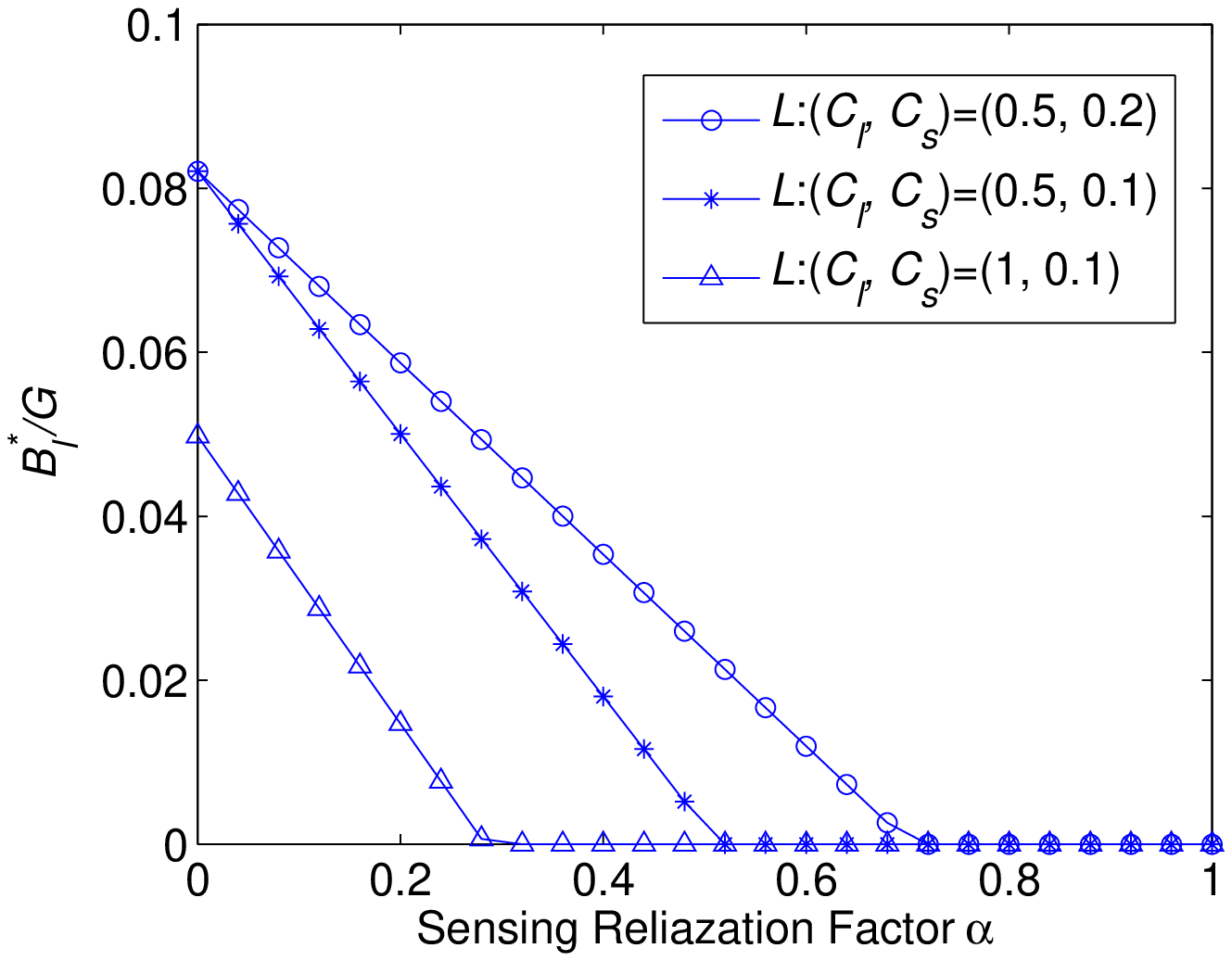}
      \caption{Optimal leasing amount $B_l^*$ as a function of\com{ sensing
cost} $C_{s}$, \com{leasing cost }$C_{l}$, and\com{ sensing
realization factor} $\alpha$.} \label{fig_sim2}
   \end{minipage}
   \begin{minipage}[t]{0.32\linewidth}
      \centering
      \includegraphics[width=1.1\textwidth]{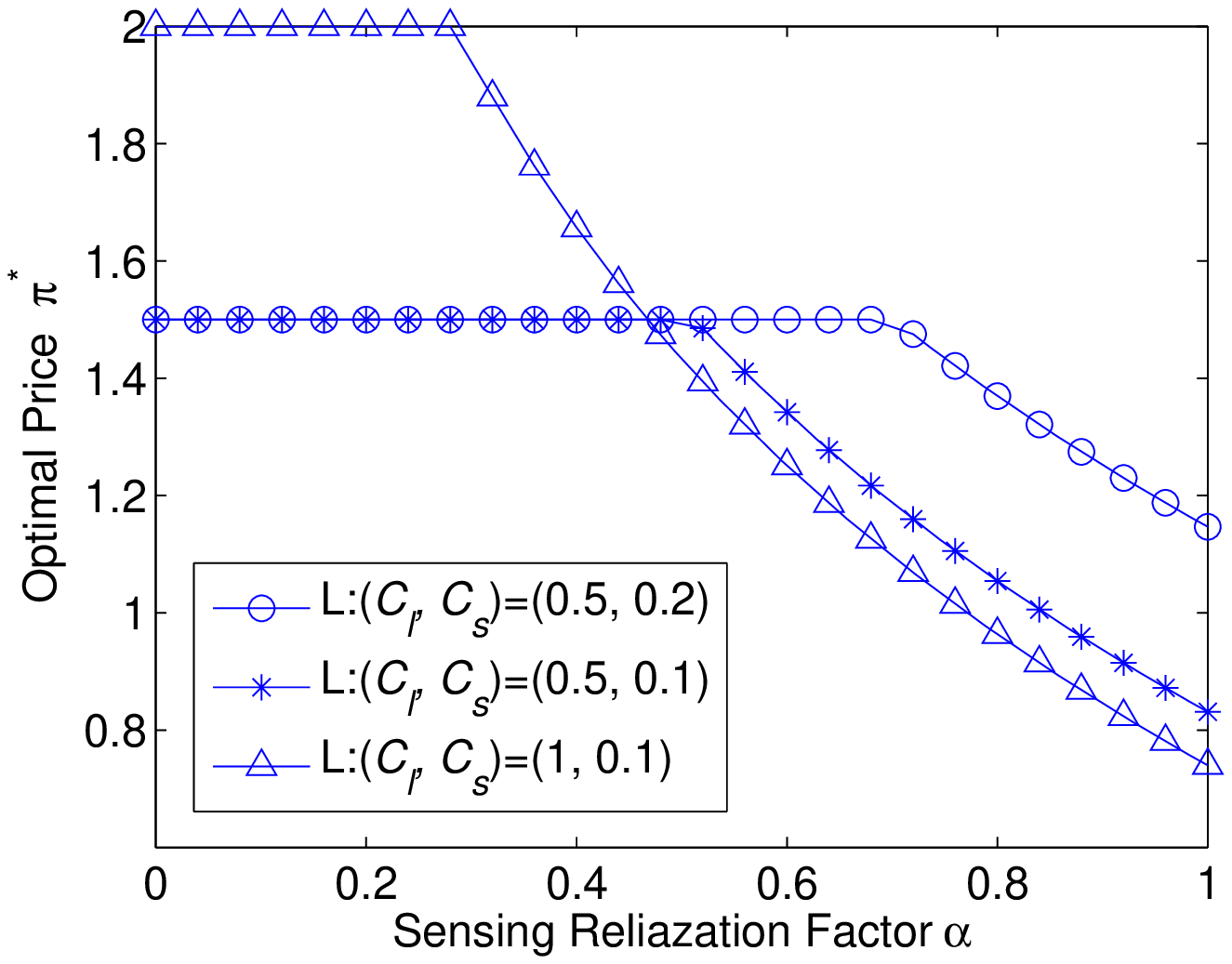}
      \caption{Optimal price $\pi^*$ as a function of \com{sensing cost
}$C_{s}$, \com{leasing cost }$C_{l}$, and \com{sensing realization
factor }$\alpha$.}
      \label{fig_sim3}
   \end{minipage}
\end{figure*}

\com{
\begin{figure}[tt]
\centering
\includegraphics[width=0.35\textwidth]{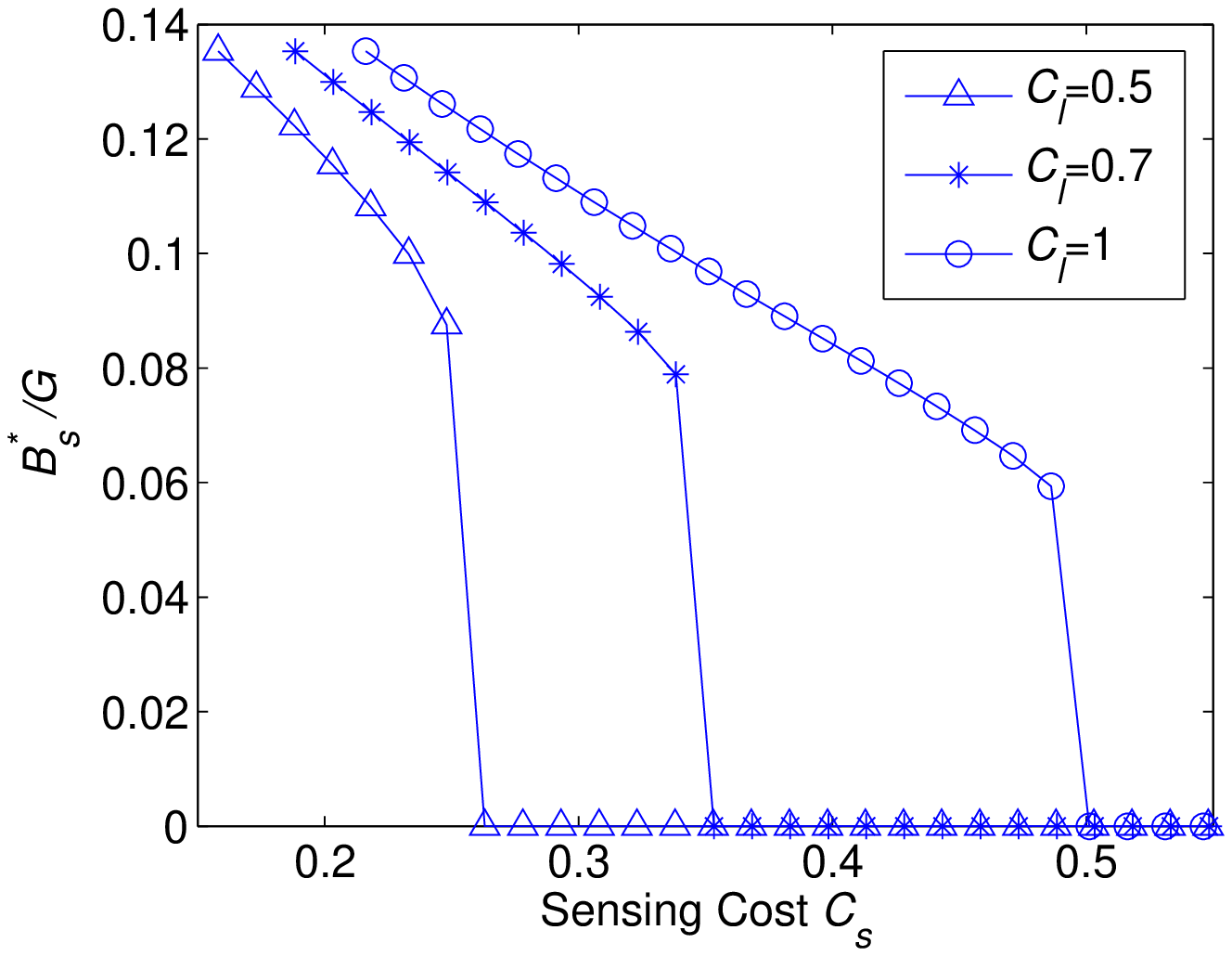}
\caption{Optimal sensing amount $B_s^*$ as a function of sensing
cost $C_{s}$ and leasing cost $C_{l}$.} \label{fig_sim1}
\end{figure}
}

\com{
\begin{figure}[tt]
\centering
\includegraphics[width=0.35\textwidth]{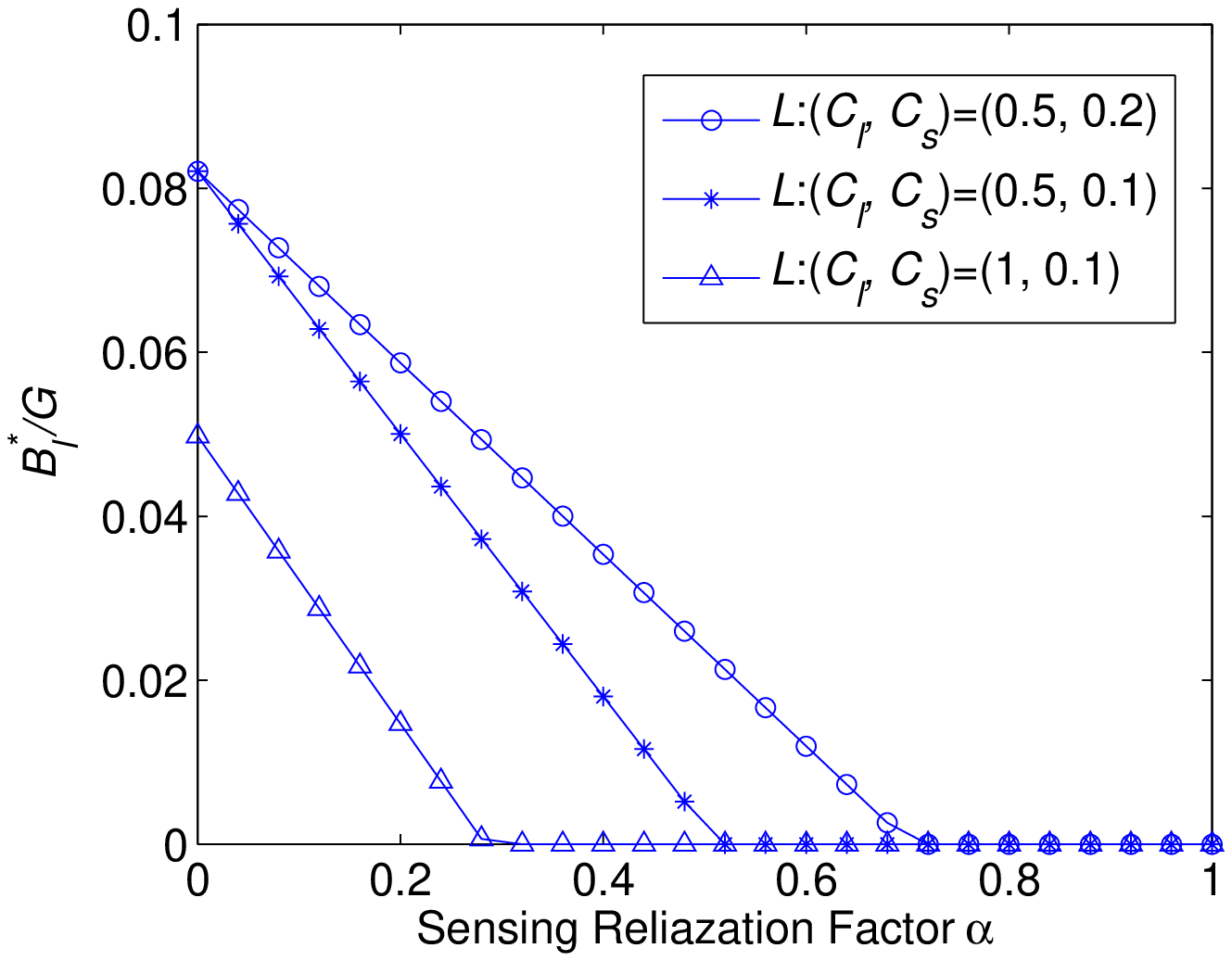}
\caption{Optimal leasing amount $B_l^*$ as a function of sensing
cost $C_{s}$, leasing cost $C_{l}$, and sensing realization factor
$\alpha$.} \label{fig_sim2}
\end{figure}
}


Based on the discussions in Section \ref{sec:BackwardInduction}, we
summarize the operator's equilibrium sensing/leasing/pricing
decisions and the equilibrium resource allocations to the users in
Table \ref{tab:equilibrium}.
%
Several interesting observations are as follows.

\begin{observation}\label{ob1}
Both the optimal sensing amount $B_s^*$ (either 0 or $B_s^{L*}$) and
leasing amount $B_l^*$ are linear in the users' aggregate wireless
characteristics $G=\sum_{i\in\mathcal{I}} {P_i^{\max} h_i}/{n_0}$.
\end{observation}

The linearity enables us to normalize optimal sensing and leasing
decisions by users' aggregate wireless characteristics, and study
the relationships between the normalized optimal decisions and other
system parameters as in Figs. \ref{fig_sim1} and \ref{fig_sim2}.


%

Figure \ref{fig_sim1} shows how the normalized optimal sensing
decision $B_s^{\ast}/G$ changes with the costs. For a given leasing
cost $C_{l}$, the optimal sensing decision $B_{s}^{\ast}$ decreases
as the sensing cost  $C_{s}$ becomes more expensive, and drops to
zero when $C_s\geq C_l/2$. For a given sensing cost $C_{s}$, the
optimal sensing decision $B_{s}^{\ast}$ increases as the leasing
cost $C_{l}$ becomes more expensive, in which case sensing becomes
more attractive.





Figure \ref{fig_sim2} shows how the normalized optimal leasing
decision $B_l^{\ast}/G$ depends on the costs $C_l$ and $C_s$ as well
as the sensing realization factor $\alpha$ in the low sensing cost
regime (denoted by ``$L$''). In all cases, a higher value $\alpha$
means more bandwidth is obtained from sensing and there is a less
need to lease. Figure \ref{fig_sim2} confirms the threshold
structure of the optimal leasing decisions in Section
\ref{Sect:stageII}, i.e., no leasing is needed whenever the
bandwidth obtained from sensing reaches a threshold. Comparing
different curves, we can see that the operator chooses to lease more
as leasing becomes cheaper or sensing becomes more expensive. For
high sensing cost regime, the optimal leasing amount only depends on
$C_l$ and is independent of $C_{s}$ and $\alpha$, and thus is not
shown here.

\begin{observation} \label{ob2}
The optimal pricing decision $\pi^\ast$ in Stage III is independent
of users' aggregate wireless characteristics $G$.
\end{observation}

Observation \ref{ob2} is closely related to Observation \ref{ob1}.
Since the total bandwidth  is linear in $G$, the ``average''
resource allocation per user is ``constant'' at the equilibrium.
This implies that the price must be independent of the user
population change, otherwise the resource allocation to each
individual user will change with the price accordingly.

\begin{observation}\label{ob3}
The optimal pricing decision $\pi^\ast$ in Stage III is
non-increasing in $\alpha$  in the low sensing cost regime.
\end{observation}
%

First, in the low sensing cost regime where the sensing result is
poor (i.e., $\alpha$ is small as the third column in Table
\ref{tab:equilibrium}), the operator will lease additional resource
such that the total bandwidth reaches the threshold
$Ge^{-(2+C_{l})}$. In this case, the price is a constant and is
independent of the value of $\alpha$. Second, when the sensing
result is good (i.e., $\alpha$ is large as in the last column in
Table \ref{tab:equilibrium}), the total bandwidth is large enough.
In this case, as  $\alpha$ increases,  the amount of total bandwidth
increases, and the optimal price decreases to maximize the profit.

Figure~\ref{fig_sim3} shows how the optimal price changes with
various costs and $\alpha$ in the low sensing cost regime. It is
clear that price is first a constant and then starts to decrease
when $\alpha$ is larger than a threshold. The threshold decreases in
the optimal sensing decision of $B_{s}^{L\ast}$: a smaller sensing
cost or a higher leasing cost will lead to a higher $B_{s}^{L\ast}$
and thus a smaller threshold.

\com{
\begin{figure}[tt]
\centering
\includegraphics[width=0.32\textwidth]{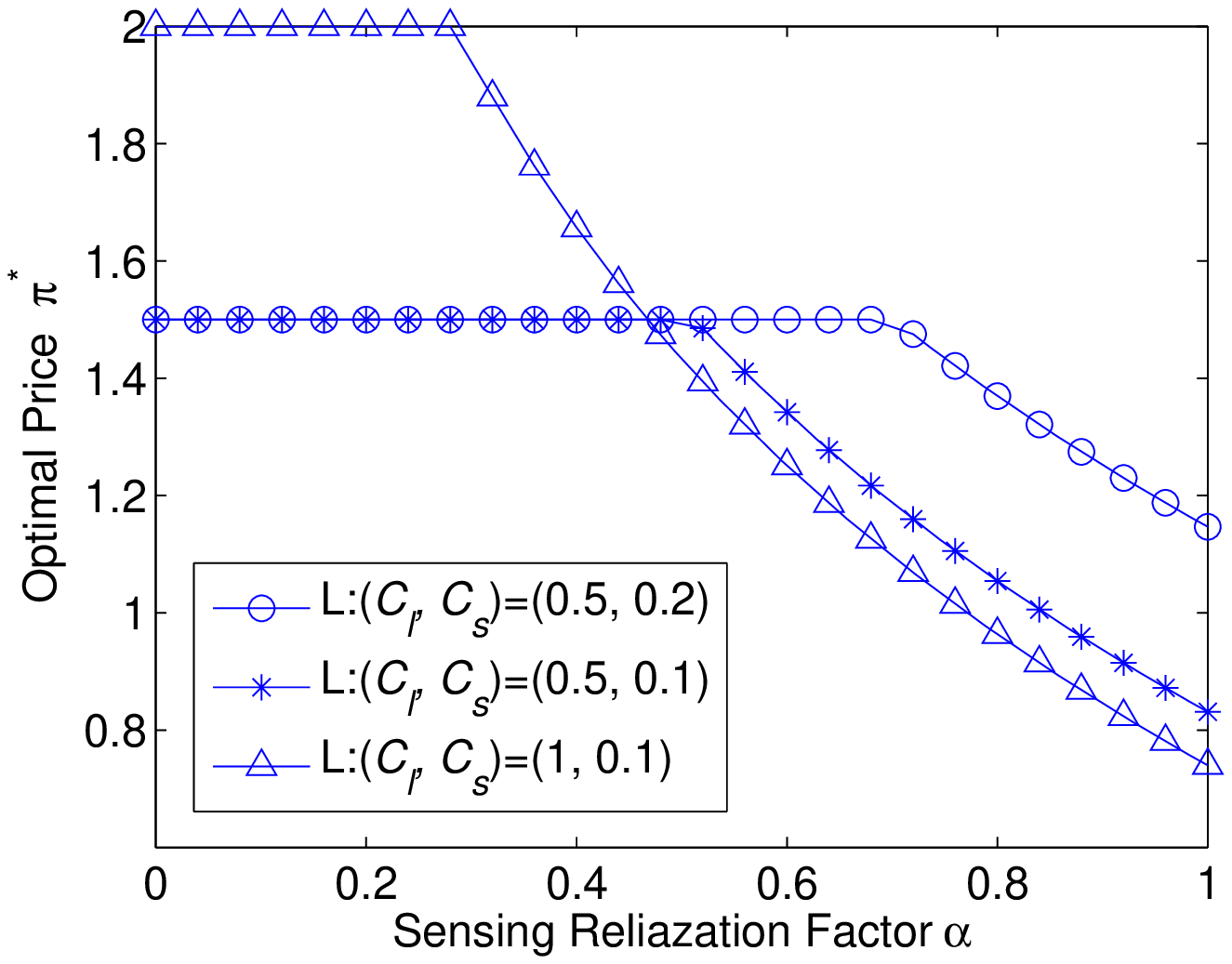}
\caption{Optimal price $\pi^*$ as a function of sensing cost
$C_{s}$, leasing cost $C_{l}$, and sensing realization factor
$\alpha$.} \label{fig_sim3}
\end{figure}
}

\begin{figure}[tt]
\centering
\includegraphics[width=0.53\textwidth]{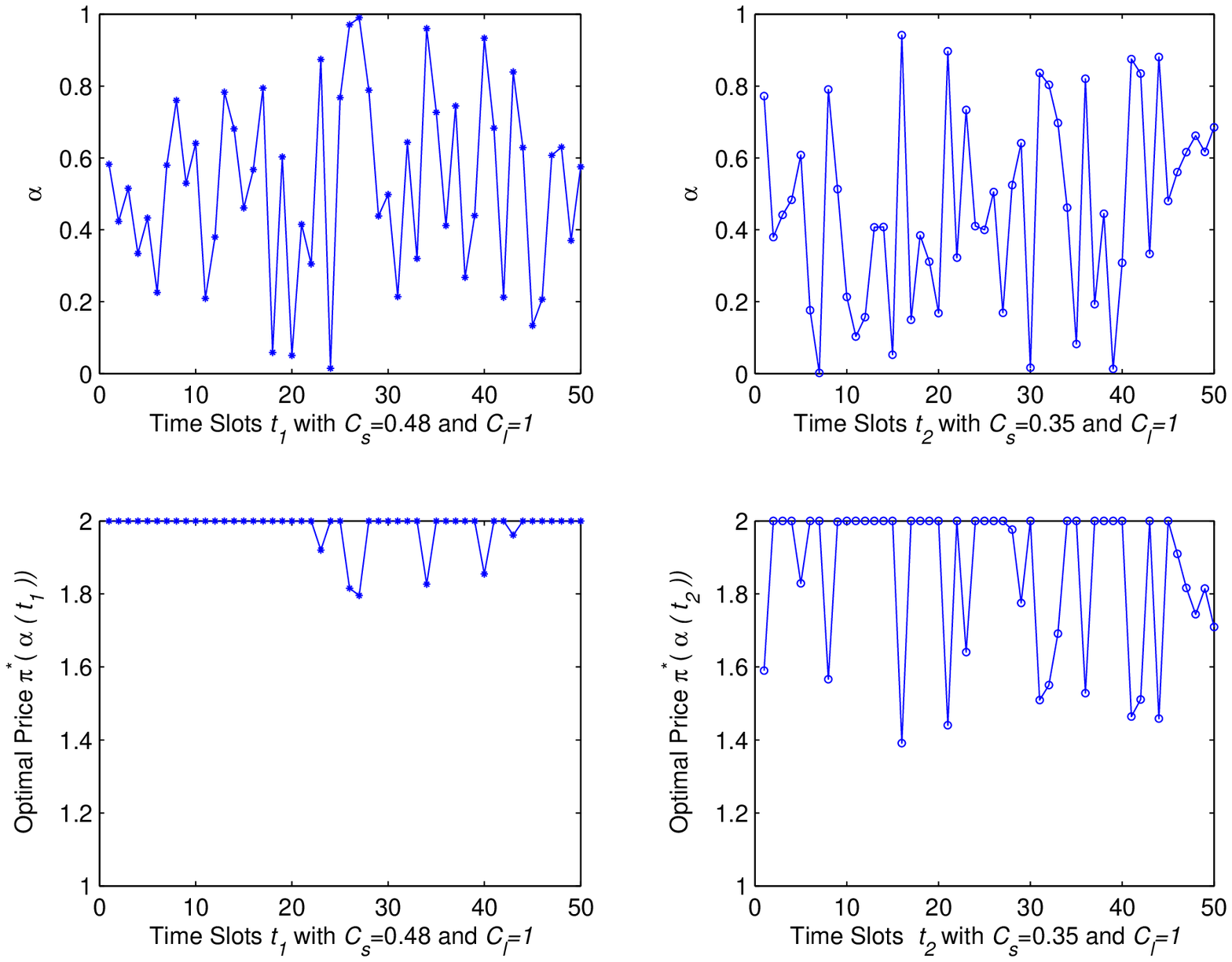}
\caption{{Optimal price $\pi^*$ over time with different sensing
costs and $\alpha$ realizations}} \label{fig:PriceDynamics}
\end{figure}

{It is interesting to notice that the equilibrium price only changes
in a time slot where the sensing realization factor $\alpha$ is
large. This means that although operator has the freedom to change
the price in every time slot, the actual variation of price is much
less frequent. This leads to less overhead and makes it easier to
implement in practice. Figure~\ref{fig:PriceDynamics} illustrates
this with different sensing costs  and $\alpha$ realizations. The
left two subfigures correspond to the realizations of $\alpha$ and
the corresponding prices with $C_s= 0.48$ and $C_l=1$. As the
sensing cost $C_{s}$ is quite high in this case, the operator does
not rely heavily on sensing. As a result, the variability of
$\alpha$ (in the upper subfigure) has very small impact on the
equilibrium price (in the lower subfigure). In fact, the price only
changes in 11 out 50 time slots, and the maximum amplitude variation
is around $10\%$. The right two figures correspond to the case where
$C_s= 0.35$ and $C_l=1$. As sensing cost is cheaper in this case,
the operator senses more and the impact of $\alpha$ on price is
larger. The price changes in 30 out of 50 time slots, and the
variation in amplitude can be as large as $30\%$.}


\begin{observation}\label{ob4}
The operator will sense the spectrum only if the sensing cost is
lower than a threshold. Furthermore, it will lease additional
spectrum only if the spectrum obtained through sensing is below a
threshold.
\end{observation}


\begin{observation}\label{ob5}
Each user $i$ obtains the same SNR independent of $g_{i}$ and a
payoff linear in $g_{i}$.
\end{observation}

Observation \ref{ob5} shows that users obtains fair and predictable
resource allocation at the equilibrium. In fact, a user does not
need to know anything about the total number and payoffs of other
users in the system. It can simply predict its QoS  if it knows the
cost structure of the network ($C_{s}$ and $C_{l}$)\footnote{The
analysis of the game, however, does not require the users to know
$C_s$ or $C_l$.}. Such property is highly desirable in practice.


Finally, users achieve the same high SNR at the equilibrium. The SNR
value is either $e^{(2+C_l)}$ or $G/(B_s^{L*}\alpha)$, both of which
are larger than $e^2$. This means that the approximation ratio
$\ln(\mathtt{SNR}_i)/\ln(1+\mathtt{SNR}_i)>\ln(e^{2})/\ln(1+e^{2})\approx94\%$.
The ratio can even be close to one if the price $\pi$ is high.

In Sections \ref{subsec:stageIV} and \ref{subsec:stageI}, we made
the high SNR regime approximation and the uniform distribution
assumption of $\alpha$ to obtain closed-form expressions. Next we
show that relaxing both assumptions will not change any of the major
insights.

\subsection{Robustness of the Observations}\label{subsec:robustness}
\begin{theorem}\label{thm:robustness}
Observations \ref{ob1}-\ref{ob5} still hold under the general SNR
regime (as in (\ref{eq:rate}))  and any general distribution of
$\alpha$.
\end{theorem}



\begin{proof}
We represent a user $i$'s payoff function in the general SNR regime,
\begin{equation}\label{eq:generalSNR}
u_i(\pi,w_i) = w_i \ln\left(1+\frac{g_i}{w_i}\right)-\pi w_i.
\end{equation}
The optimal demand $w_i^*(\pi)$ that maximizes (\ref{eq:generalSNR})
is $w_i^*(\pi)=g_i/Q(\pi),$ where $Q(\pi)$ is the unique positive
solution to $F(\pi,Q):=\ln(1+Q)-\frac{Q}{1+Q}-\pi=0.$ We find the
inverse function of $Q(\pi)$ to be $\pi(Q)=\ln(1+Q)-\frac{Q}{1+Q}$.
By applying the implicit function theorem, we can obtain the
first-order derivative of function $Q(\pi)$ over $\pi$ as
\begin{equation}\label{eq:dQ(pi)}Q'(\pi)=-\frac{\partial F(\pi,Q)/\partial \pi}{\partial
F(\pi,Q)/\partial Q}=\frac{(1+Q(\pi))^2}{Q(\pi)},\end{equation}
which is always positive. Hence, $Q(\pi)$ is increasing in $\pi$.

User $i$'s optimal payoff is
\begin{equation}\label{eq:generalpayoff}u_i(\pi,w_i^*(\pi))=\frac{g_i}{Q(\pi)}[\ln(1+Q(\pi))-\pi].\end{equation}
As a result, a user's optimal SNR equals
$g_{i}/w_{i}^{\ast}(\pi)=Q(\pi)$ and is \emph{user-independent}. The
total demand from all users equals ${G}/Q(\pi)$, and the operator's
 investment and pricing problem is
\begin{align}\label{eq:star}
R^*=&\max_{B_s\geq 0} E_{\alpha\in [0,1]} [\max_{B_l\geq0}
\max_{\pi\geq 0}
(\min\left(\pi\frac{G}{Q(\pi)},\pi(B_l+B_s\alpha)\right)\nonumber\\&-B_sC_s-B_lC_l)].
\end{align}
Define $\widetilde{R^*}=\frac{R^*}{{G}},
\widetilde{B_l}=\frac{B_l}{{G}}$, and
$\widetilde{B_s}=\frac{B_s}{{G}}$.
Then solving (\ref{eq:star})  is equivalent to solving
\begin{align}\label{eq:starstar}
\widetilde{R^*}=&\max_{\widetilde{B_s}\geq 0} E_{\alpha\in [0,1]}
[\max_{\widetilde{B_l}\geq0} \max_{\pi\geq 0}
(\min\left(\frac{\pi}{Q(\pi)},\pi(\widetilde{B_l}+\widetilde{B_s}\alpha)\right)\nonumber\\&-\widetilde{B_s}C_s-\widetilde{B_l}C_l)].
\end{align}
In Problem (\ref{eq:starstar}), it is clear that the operator's
optimal decisions on leasing, sensing and pricing do not depend on
users' aggregate wireless characteristics. This is true for any
continuous distribution of $\alpha$. And a user's optimal payoff in
eq. (\ref{eq:generalpayoff}) is linear in $g_i$ since $Q(\pi)$ is
independent of users' wireless characteristics. This shows that
Observations \ref{ob1}, \ref{ob2}, and \ref{ob5} hold  for the
general SNR regime and any general distribution of $\alpha$. We can
also show that Observations \ref{ob3} and \ref{ob4} hold in the
general case, with a detailed proof in Appendix \ref{app:thm4}.
\end{proof}

%

%



\section{The Impact of Spectrum Sensing Uncertainty}
\label{sec:SensingImpact}
%

The key difference between our model and most existing literature
(e.g.,
\cite{IEEEhowto:Jia,IEEEhowto:Niyato,jia2009revenue,sengupta2007economic,IEEEhowto:Ileri,IEEEhowto:Xing})
is the possibility of obtaining resource through the cheaper but
uncertain approach of spectrum sensing. Here we will elaborate the
impact of sensing on the performances of operator and users by
comparing with the \emph{baseline case} where sensing is not
possible. Note that in the high sensing cost regime it is optimal
not to sense, as a result, the performance of the operator and users
will be the same as the baseline case. Hence we will focus on the
low sensing cost regime in Table \ref{tab:equilibrium}.

\begin{observation}\label{prob1}
The operator's optimal \emph{expected} profit always benefits from
the availability of spectrum sensing in the low sensing cost regime.
\end{observation}

\begin{figure*}
   \begin{minipage}[t]{0.32\linewidth}
      \centering
      \includegraphics[width=1.1\textwidth]{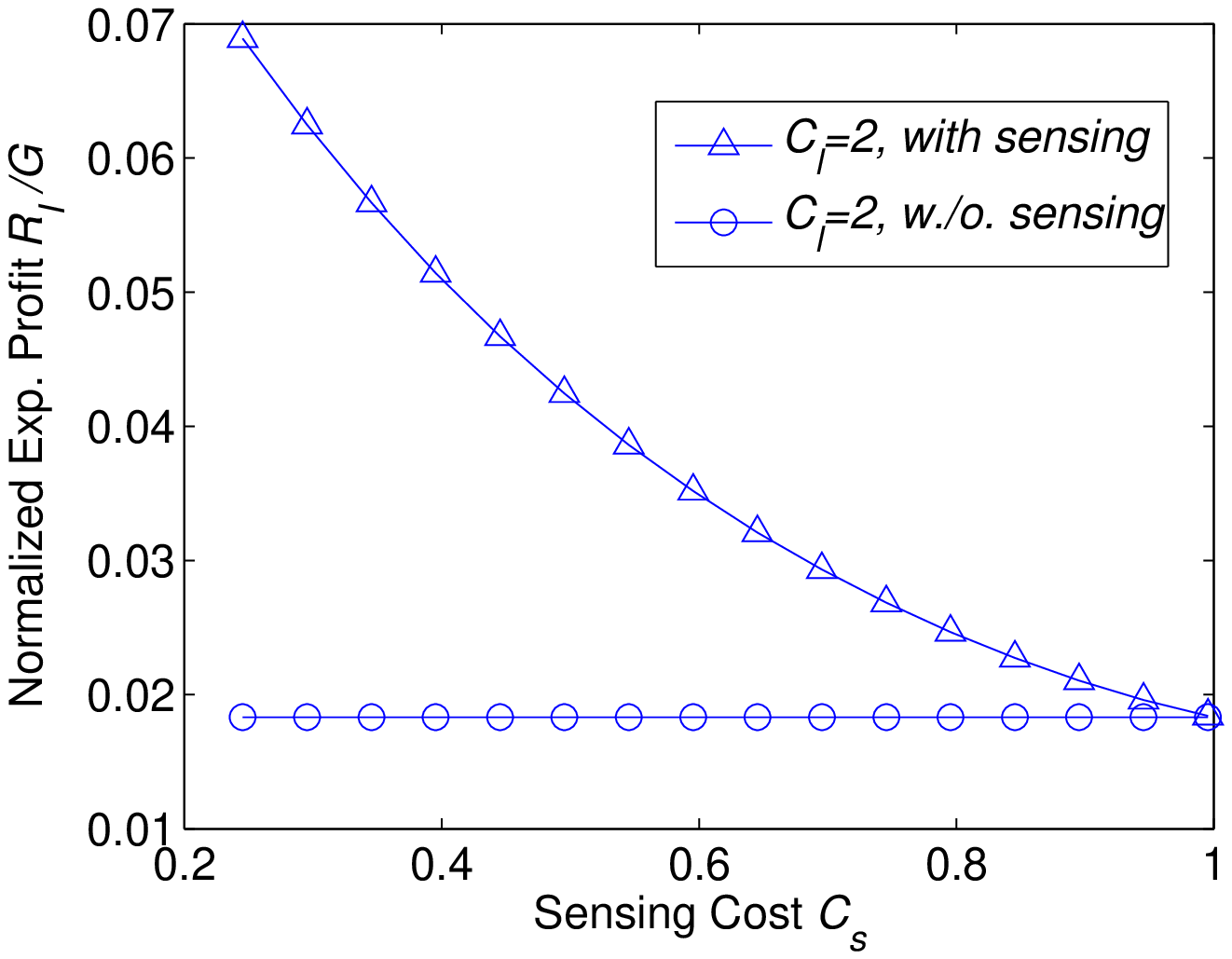}
\caption{Operator's normalized optimal \rev{\emph{expected}} profit
as a function of \com{sensing cost }$C_{s}$ and \com{leasing cost
}$C_{l}$.\com{ The baseline is the case without sensing.}}
\label{fig:comparingrevenueCs}
   \end{minipage}
   \begin{minipage}[t]{0.32\linewidth}
      \centering
      \includegraphics[width=1.1\textwidth]{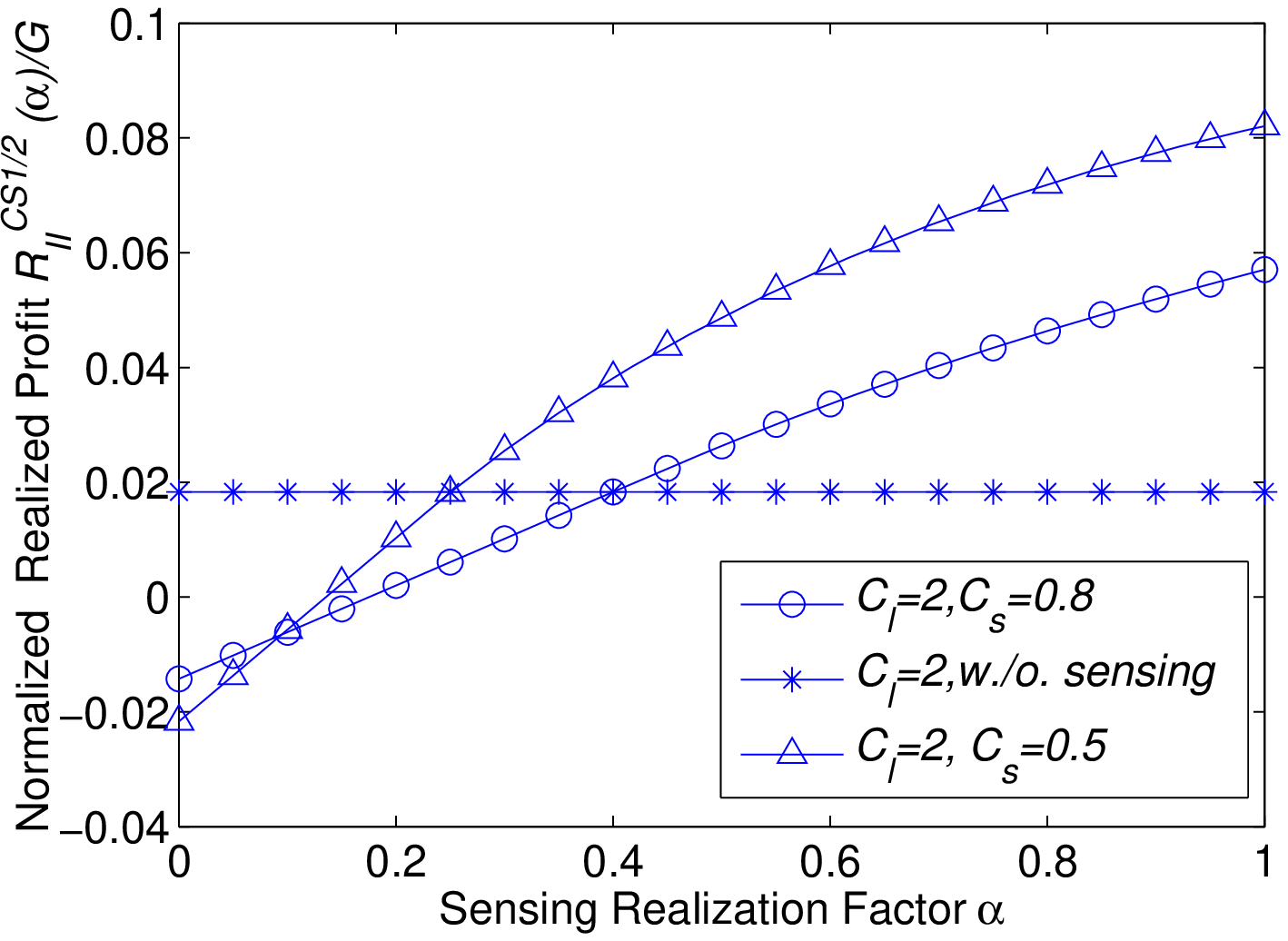}
\caption{Operator's normalized optimal \rev{\emph{realized}} profit
as a function of \com{sensing realization factor }$\alpha$. \com{The
baseline is the case without sensing.}}
\label{fig:comparingrevenuealpha}
   \end{minipage}
   \begin{minipage}[t]{0.32\linewidth}
      \centering
      \includegraphics[width=1.1\textwidth]{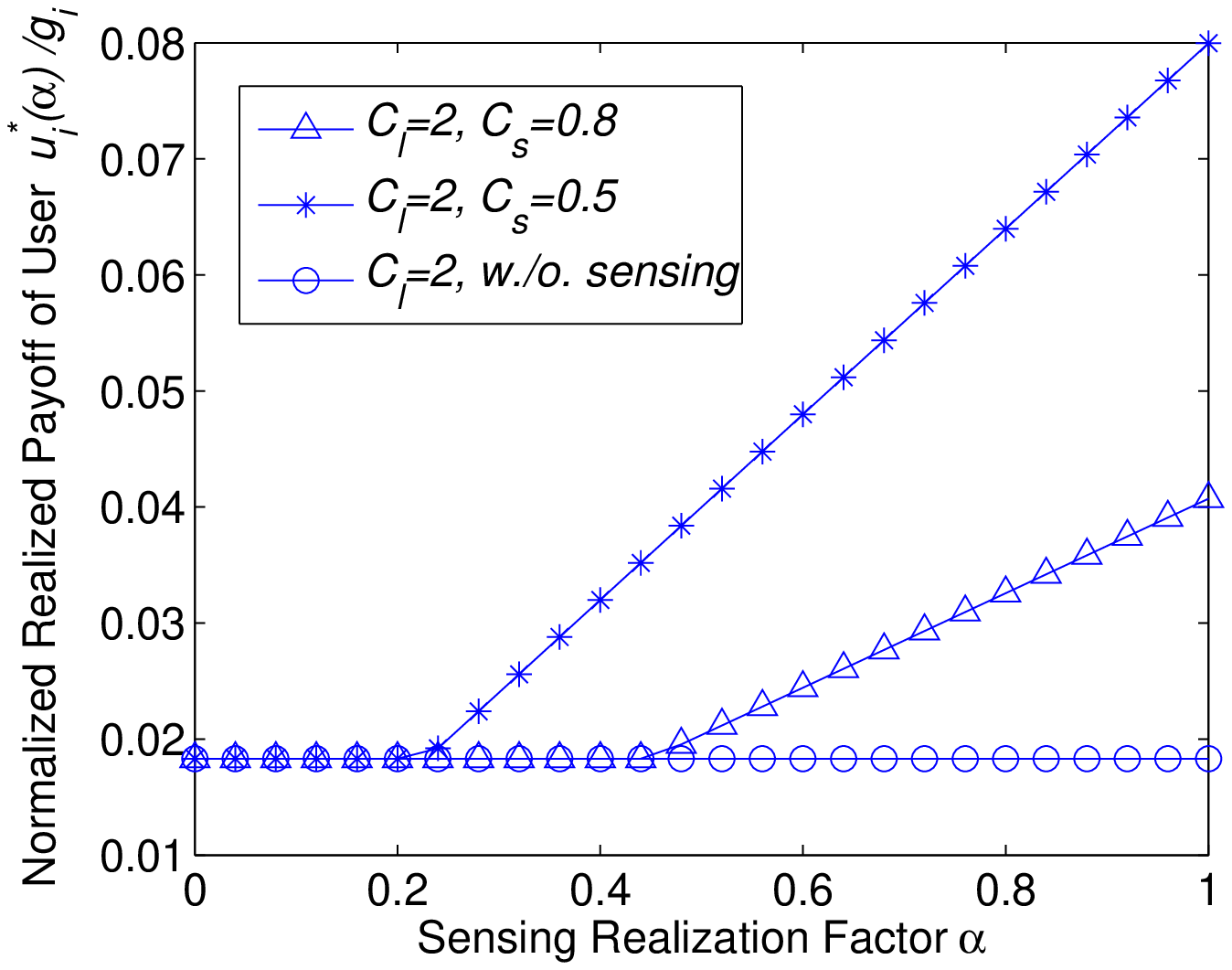}
\caption{User $i$'s normalized optimal realized payoff as a function
of \com{sensing realization factor }$\alpha$.}
\label{fig:userpayoffalpha}
   \end{minipage}
\end{figure*}

\com{
\begin{figure}[tt]
\centering
\includegraphics[width=0.34\textwidth]{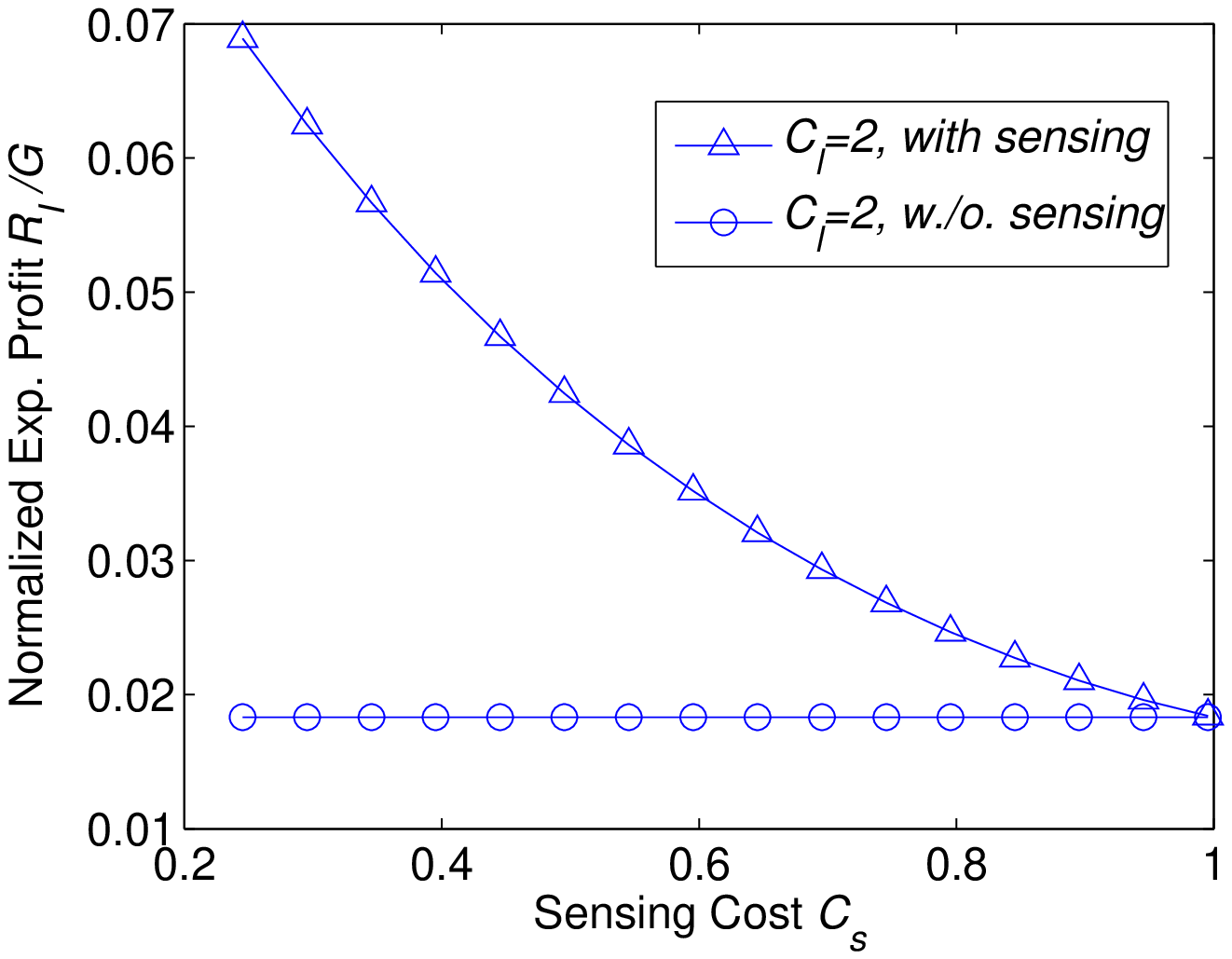}
\caption{Operator's normalized optimal expected profit as a function
of sensing cost $C_{s}$ and leasing cost $C_{l}$. The baseline is
the case without sensing.} \label{fig:comparingrevenueCs}
\end{figure}
}

Figure \ref{fig:comparingrevenueCs} illustrates the normalized
optimal expected profit as a function of the sensing cost. \com{The
baseline is the case without sensing.}We assume leasing cost
$C_l=2$, and thus the low sensing cost regime corresponds to the
case where $C_s\in[0.2,1]$ in the figure. It is clear that sensing
achieves a better optimal expected profit in this regime. In fact,
sensing leads to $250\%$ increase in profit when $C_{s}=0.2$. The
benefit decreases as the sensing cost becomes higher. When sensing
becomes too expensive, the operator will choose not to sense and
thus achieve the same profit as in the baseline case.


\begin{theorem}\label{prob2}
The operator's realized profit (i.e., the profit for a given
$\alpha$) is a strictly increasing function in $\alpha$ in the low
sensing cost regime. Furthermore, there exists a threshold
$\alpha_{th}\in(0,1)$ such that the operator's realized profit is
larger than the baseline approach if $\alpha>\alpha_{th}$.
\end{theorem}


\begin{proof}
As in Table \ref{tab:equilibrium}, we have
two cases in the low sensing cost regime: 
%
\begin{itemize}
  \item If  $\alpha\leq G e^{-(2+C_l)}/B_s^{L*}$, then
substituting $B_s^{L*}$ into $R_{II}^{CS1}(B_s,\alpha)$ in Table
\ref{tab:leasing} leads to the realized profit
$$
R_{II}^{CS1}(\alpha)=Ge^{-(2+C_l)}-B_s^{L*}C_s+B_s^{L*}\alpha C_l,$$
which is strictly and linearly increasing in $\alpha$.

\item If $\alpha\geq Ge^{-(2+C_l)}/B_s^{L*}$,
then substituting $B_s^{L*}$ into $R_{II}^{CS2}(B_s,\alpha)$ in
Table \ref{tab:leasing} leads to the realized profit
$$
R_{II}^{CS2}(\alpha)=B_s^{L*}\alpha\left(\ln\left(\frac{G}{B_s^{L*}\alpha}\right)-1\right)-B_s^{L*}
C_s.$$ Because the first-order derivative
$$\frac{\partial R_{II}^{CS2}(\alpha)}{\partial \alpha}=B_s^{L*}\left(\ln\left(\frac{G}{B_s^{L*}\alpha}\right)-2\right)>0,$$
as $B_s^{L*}\leq Ge^{-2}$, $R_{II}^{CS2}(\alpha)$ is strictly
increasing in $\alpha$.
\end{itemize}
We can also verify that $R_{II}^{CS1}(\alpha)=R_{II}^{CS2}(\alpha)$
when $\alpha=G e^{-(2+C_l)}/B_s^{L*}$. Therefore, the realized
profit is a continuous and strictly increasing function of $\alpha$.

Next we prove the existence of threshold  $\alpha_{th}$. First
consider the extreme case  $\alpha=0$. Since the operator obtains no
bandwidth through sensing but still incurs some cost, the profit in
this case is lower than the baseline case. Furthermore, we can
verify that $R_{II}^{CS2}(1)>R_I^H$ in Table \ref{tab:equilibrium},
thus the realized profit at $\alpha=1$ is always larger than the
baseline case. Together with the continuity and strictly increasing
nature of the realized profit function, we have proven the existence
of threshold of $\alpha_{th}$.
\end{proof}


\com{
\begin{figure}[tt]
\centering
\includegraphics[width=0.34\textwidth]{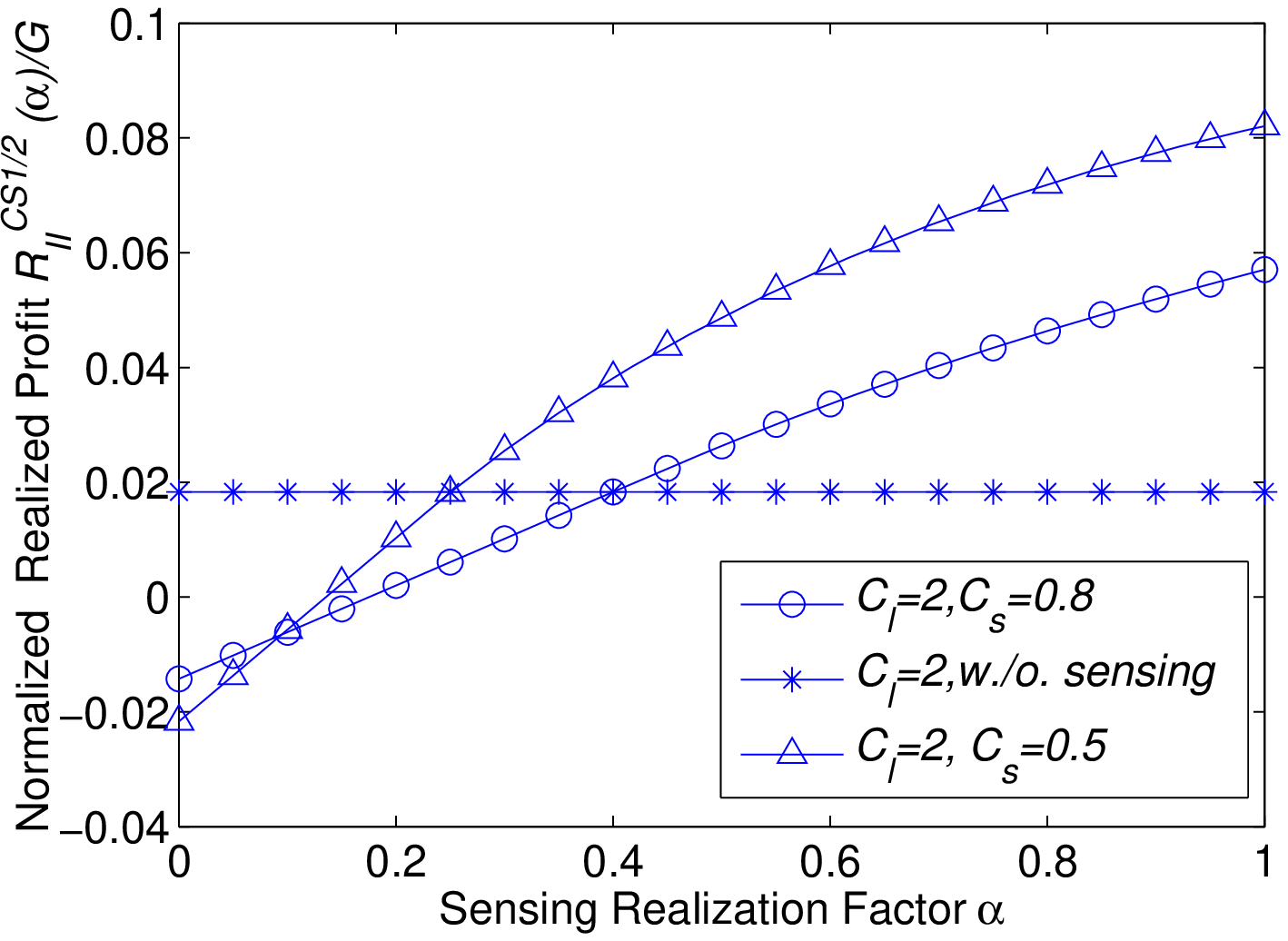}
\caption{Operator's normalized optimal realized profit as a function
of the sensing realization factor $\alpha$. The baseline is the case
without sensing.} \label{fig:comparingrevenuealpha}
\end{figure}
}

Figure \ref{fig:comparingrevenuealpha} shows the realized profit as
a function of $\alpha$ for different costs.  The realized profit is
increasing in $\alpha$ in both cases. The ``crossing'' feature of
the two increasing curves is because the optimal sensing $B_s^\ast$
is larger under a cheaper sensing cost ($C_s=0.5$), which leads to
larger realized profit loss (gain, respectively) when $\alpha
\rightarrow 0$ ($\alpha\rightarrow 1$, respectively).
This shows the tradeoff between improvement of expected profit and
the large variability of the realized profit.


\com{
\begin{figure}[tt]
\centering
\includegraphics[width=0.3\textwidth]{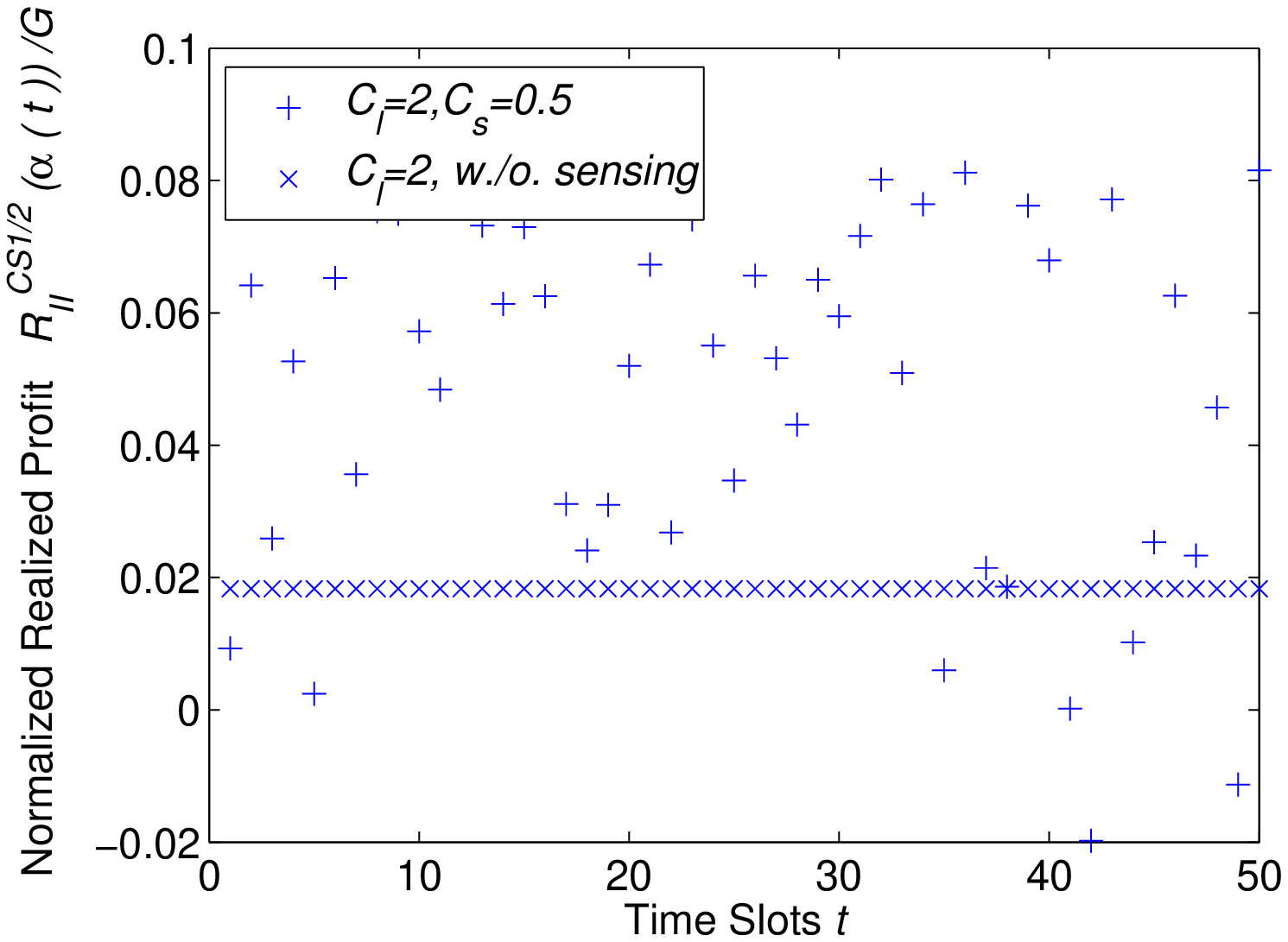}
\caption{Operator's normalized optimal realized profit as a function
of the sensing realization factor $\alpha$ in different time slots.}
\label{fig:comparingrevenuetime}
\end{figure}
}

\com{Figure \ref{fig:comparingrevenuetime} shows how the operator's
realized profit changes over time as $\alpha$ changes randomly in
each time slot according to a uniform distribution.  It is clear
that sensing leads to an increase in the expected (time average)
profit.}

%
%

\begin{theorem}\label{prob3}
Users always benefit from the availability of spectrum sensing in
the low sensing cost regime.
\end{theorem}

\begin{proof}
In the baseline approach without sensing, the operator always
charges the price $1+C_l$. As shown in Table \ref{tab:equilibrium},
the equilibrium price $\pi^\ast$ with sensing is always no larger
than $1+C_l$ for any value of $\alpha$. Since a user's payoff is
strictly decreasing in price,  the users always benefit  from
sensing.
\end{proof}



\com{
\begin{figure}[tt]
\centering
\includegraphics[width=0.33\textwidth]{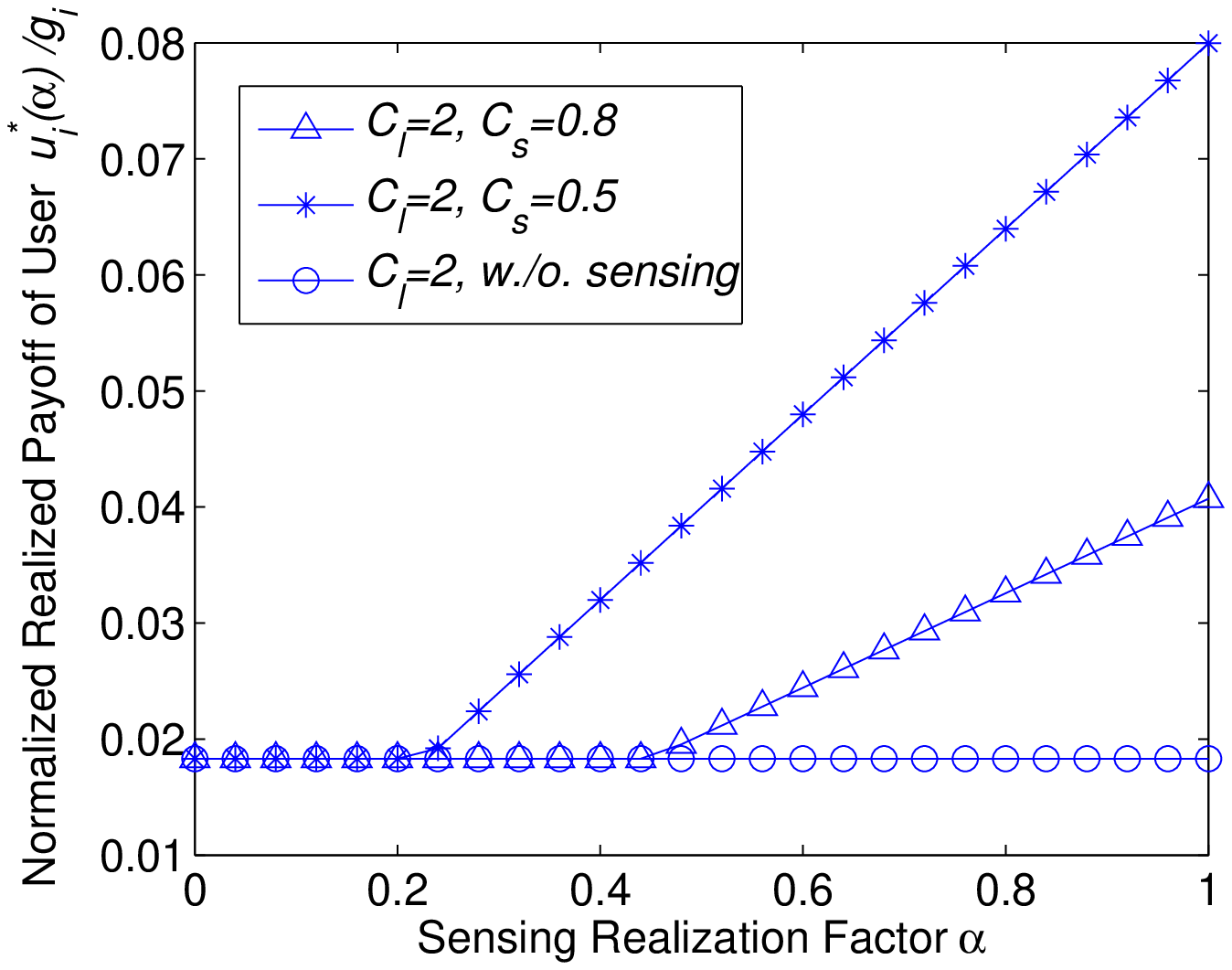}
\caption{User $i$'s normalized optimal realized payoff as a function
of the sensing realization factor $\alpha$.}
\label{fig:userpayoffalpha}
\end{figure}
}

Figure \ref{fig:userpayoffalpha} shows how a user $i$'s normalized
realized payoff $u_i^*/g_i$ changes with $\alpha$. The payoff
linearly increases in $\alpha$ when $\alpha$ becomes larger than a
threshold, in which case the equilibrium price becomes lower than
$1+C_l$. A smaller sensing cost $C_s$ leads to more aggressive
sensing and thus more benefits to the users.

\section{Conclusions and Future Work}
\label{sec:conclusion}

This paper represents {some initial results}  towards understanding
the new business models, opportunities, and challenges of the
emerging cognitive virtual mobile network operators (C-MVNOs) {under
supply uncertainty}. Here we focus on studying the trade-off between
the cost and uncertainty of spectrum investment through sensing and
leasing.
We model the interactions between the operator and the users by a
{Stackelberg game}, which captures the wireless heterogeneity of
users in terms of maximum transmission power levels and channel
gains.

We have discovered several interesting features of the game
equilibrium. We show that the operator's optimal sensing, leasing,
and pricing decisions follow nice threshold structures. The
availability of sensing always increases the operator's expected
profit, despite that the realized profit in each time slot will have
some variations depending on the sensing result. Moreover, users
always benefit in terms of payoffs when sensing is performed by the
operator.

{To keep the problem tractable, we have made several assumptions
throughout this paper. Some assumptions can be (easily) generalized
without affecting the main insights.
\begin{itemize}
\item \emph{Imperfect sensing}:  we can
incorporate imperfect spectrum sensing (i.e., miss-detection and
false-positive) into the model, which will change the uncertainty of
the spectrum sensing. Given that our results work for any
distribution of the sensing realization $\alpha$, it is likely that
such generalization does not change the major insights.

\item \emph{Learning}: we can also consider the
interactions of multiple time slots, where the sensing realizations
of previous time slots can be used to update the distributions of
$\alpha$ in future time slots. Again, the per slot decision model
introduced in this paper is still applicable with a time-dependent
$\alpha$ distribution input.
\end{itemize}

Generalizations of some other assumptions, however, lead to more
challenging new problems.
\begin{itemize}
\item \emph{Incomplete information}:  when the operator does not know the
information of the users, the system needs to be modeled as a
dynamic game with incomplete information. More elaborate economic
models such as screening and signaling \cite{rasmusen2007games}
become relevant.

\item \emph{Time scale separation}: it is possible that dynamic leasing is performed at a different (much larger) time scale compared with spectrum sensing. In  that case, the operator has to make the leasing decision first, and then make several sequential sensing decisions. This leads to a dynamic decision model with more stages and tight couplings across sequential decisions.

\item \emph{Operator competition}: There may be multiple C-MVNOs providing services in the same geographic area. In that case, the operators need to attract the users through price competition. Also, if they sense and lease from the same spectrum owner, the operators may have overlapping or conflicting resource requests. Although we have obtained some preliminary results along this line in \cite{ourDySPAN}, more studies are definitely desirable.
\end{itemize}

Through the analytical and simulation study of an idealized model in
this paper, we have obtained various interesting engineering and
economical insights into the operations of C-MVNOs. We hope that
this paper can contribute to the further understanding of proper
network architecture decisions and business models of future
cognitive radio systems. }

\appendices
\section{Proof of Theorem \ref{Thm:pricing}}\label{app:thm1} Given
the total bandwidth $B_l+B_s\alpha$, the objective of Stage III is
to solve the optimization problem (\ref{eq:maxmin}), i.e.,
$\max_{\pi\geq0}\min(D(\pi),S(\pi))$.  First, by examining the
derivative of $D(\pi)$, i.e., ${\partial D(\pi)}/{\partial
\pi}=(1-\pi)Ge^{-(1+\pi)},$ we can see that the continuous function
$D(\pi)$ is increasing in $\pi\in[0,1]$ and decreasing in
$\pi\in[1,+\infty]$, and $D(\pi)$ is maximized when $\pi=1$. Since
$S(\pi)$ always increases in $\pi$ and $D(\pi)$ is concave over
$\pi\in[0,1]$,  $S(\pi)$ intersects with $D(\pi)$ if and only if
$\frac{\partial D(\pi)}{\partial \pi}>\frac{\partial
S(\pi)}{\partial \pi}$ at $\pi=0$, i.e., $B_l+B_s\alpha< Ge^{-1}$.

Next we divide our discussion into the intersection case  and the
non-intersection case:

\begin{enumerate}
\item Given $B_{l}+B_{s}\alpha\leq Ge^{-1}$, $S(\pi)$ intersects with $D(\pi)$. By solving equation
$S(\pi)=D(\pi)$ the intersection point is
$\pi=\ln\left(\frac{G}{B_{l}+B_{s}\alpha}\right)-1$. There are two
subcases:

\begin{itemize}
\item when $B_{l}+B_{s}\alpha\leq Ge^{-2}$, $S(\pi)$ intersects with $D(\pi)$, and $\min(D(\pi),S(\pi))$ is
maximized at the intersection point, i.e.,
$\pi^*=\ln\left(\frac{G}{B_{l}+B_{s}\alpha}\right)-1$. (See
$S_3(\pi)$ in Fig.~\ref{fig:intersection}.)

\item when $B_{l}+B_{s}\alpha\geq Ge^{-2}$, $S(\pi)$ intersects with $D(\pi)$, and
$\min(D(\pi),S(\pi))$ is maximized at the maximum value of $D(\pi)$,
i.e., $\pi^{*}=1$. (See $S_2(\pi)$ in Fig.~\ref{fig:intersection}.)
\end{itemize}

\item Given $B_{l}+B_{s}\alpha\geq Ge^{-1}$, $S(\pi)$ doesn't intersect with $D(\pi)$. Then
$\min(D(\pi),S(\pi))$ is maximized at the maximum value of $D(\pi)$,
i.e., $\pi^{*}=1$. (See $S_1(\pi)$ in Fig.~\ref{fig:intersection}.) \hfill\rule{2mm}{2mm}
\end{enumerate}

\section{Proof of Theorem \ref{thm:leasing}}\label{app:thm2}
Given the sensing result $B_s\alpha$, the objective of Stage II is
to solve the decomposed two subproblems (\ref{eq:leasingSub2}) and
(\ref{eq:leasingSub1}), and select the best one with better optimal
performance. Since $R_{III}^{ES}(B_s,\alpha,B_l)$ in subproblem
(\ref{eq:leasingSub2}) is linearly decreasing in $B_l$, its optimal
solution always lies at the lower boundary of the feasible set
(i.e., $B_l^*=\max\{Ge^{-2}-B_s\alpha,0\}$). We compare the optimal
profits of two subproblems (i.e., $R_{II}^{ES}(B_s,\alpha)$ and
$R_{II}^{CS}(B_s,\alpha)$) for different sensing results:

\begin{enumerate}
\item Given $B_s\alpha>Ge^{-2}$, the
obtained bandwidth after Stage I is already in excessive supply
regime. Thus it is optimal not to lease for subproblem
(\ref{eq:leasingSub2}) (i.e., $B_l^{ES3}=0$ of case (ES3) in Table
\ref{tab:leasing}).

\item Given $0\leq B_s\alpha\leq Ge^{-2}$, the optimal
leasing decision for subproblem (\ref{eq:leasingSub1}) is
$B_l^*=Ge^{-2}-B_s\alpha$ and we have
$R_{III}^{ES}(B_s,\alpha,B_l)=R_{III}^{CS}(B_s,\alpha,B_l)$ when
$B_l=Ge^{-2}-B_s\alpha$, thus the optimal objective value of
(\ref{eq:leasingSub2}) is always no larger than that of
(\ref{eq:leasingSub1}) and it is enough to consider the conservative
supply regime only. Since $$\frac{\partial^2
R_{III}^{CS}(B_s,\alpha,B_l)}{\partial
B_l^2}=-\frac{1}{B_l+B_s\alpha}<0,$$
$R_{III}^{CS}(B_s,\alpha,B_l)$ is concave in $0\leq B_l\leq G
e^{-2}-B_s\alpha$. Thus it is enough to examine the first-order
condition $$\frac{\partial R_{III}^{CS}(B_s,\alpha,B_l)}{\partial
B_l} =\ln\left(\frac{G}{B_l+B_s\alpha}\right)-2-C_l =0,$$ and the
boundary condition $0\leq B_l\leq Ge^{-2}-B_s\alpha$. This results
in optimal leasing decision $B_l^*=\max(G e^{-(2+C_l)}-B_s\alpha,0)$
and leads to $B_l^{CS1}=G e^{-(2+C_l)}-B_s\alpha$ and $B_l^{CS2}=0$
of cases (CS1) and (CS2) in Table \ref{tab:leasing}.
\end{enumerate}

By substituting $B_l^{CS1}$ and $B_l^{CS2}$ into
$R_{III}^{CS}(B_s,\alpha,B_l)$ in Table \ref{tab:pricing}, we derive
the corresponding optimal profits $R_{II}^{CS1}(B_s,\alpha)$ and
$R_{II}^{CS2}(B_s,\alpha)$ in Table \ref{tab:leasing}.
$R_{II}^{ES3}(B_s,\alpha)$ can also be obtained by substituting
$B_l^{ES3}$ into
$R_{III}^{ES}(B_s,\alpha,B_l)$.\hfill\rule{2mm}{2mm}

\section{Supplementary Proof of Theorem
\ref{thm:robustness}}\label{app:thm4}

In this section, we prove that Observations \ref{ob3} and \ref{ob4}
hold for the genera case (i.e., the general SNR regime and a general
distributions of $\alpha$). We first show that Observation \ref{ob4}
holds for the general case.

\subsection{Threshold structure of sensing} It is not difficult to show that if the
sensing cost is much larger than the leasing cost, the
operator has no incentive to sense but will directly lease. Thus the
threshold structure on the sensing decision in Stage I still holds
for the general case. We ignore the details due to space limitations.

\subsection{Threshold structure of leasing} Next we show the threshold
structure on leasing in Stage II also holds. Similar as in the proof
of Theorem \ref{Thm:pricing}, we define $D(\pi)=\pi\frac{G}{Q(\pi)}$
and $S(\pi)=\pi(B_s\alpha+B_l)$.
\begin{itemize}
\item We first show that $D(\pi)$ is increasing when $\pi\in[0,0.468]$
and decreasing when $\pi\in[0.468,+\infty)$. To see this, we take
the first-order derivative of $D(\pi)$ over $\pi$,
$$D'(\pi)=\frac{2Q(\pi)^2+Q(\pi)-(1+Q(\pi))^2\ln(1+Q(\pi))}{Q(\pi)^3},$$
which is positive when $Q(\pi)\in[0,2.163)$ and negative when
$Q(\pi)\in[2.163,+\infty)$. Since eq. (\ref{eq:dQ(pi)}) shows that
$Q(\pi)$ is increasing in $\pi$ and $\pi(Q)\mid_{Q=2.163}=0.468$, as
a result $D(\pi)$ is increasing in $\pi\in[0,0.468]$ and decreasing
in $\pi\in[0.468,+\infty)$.  In other words, $D(\pi)$ is maximized
at $\pi=0.468$.

\begin{figure}[tt]
\centering
\includegraphics[width=0.3\textwidth]{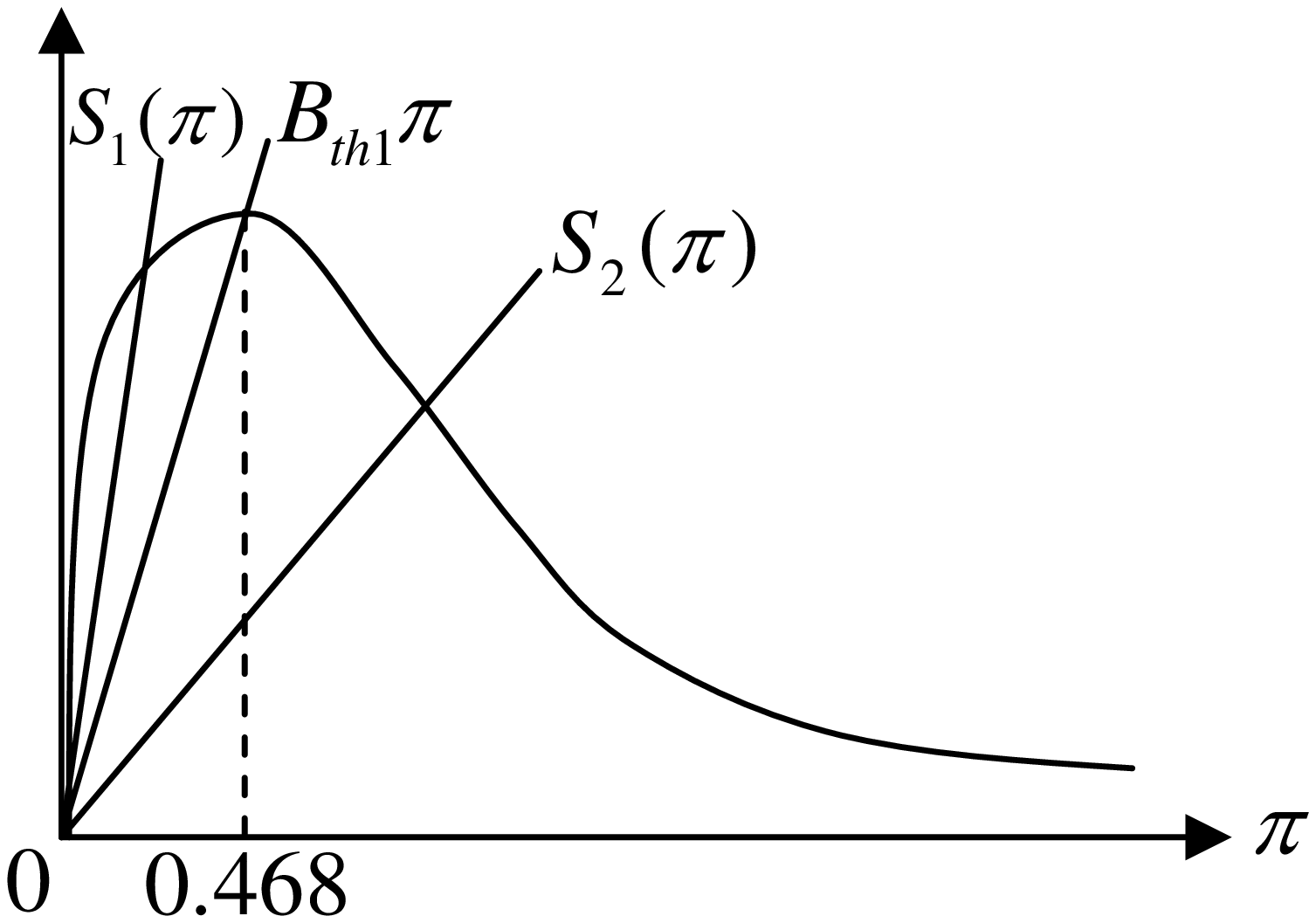}
\caption{Different intersection cases of $S(\pi)$ and $D(\pi)$ in
the general SNR regime.} \label{fig:inter_general}
\end{figure}

\item Next we derive the operator's optimal pricing decision in Stage
III. Figure \ref{fig:inter_general} shows two possible intersection
cases of $S(\pi)$ and $D(\pi)$.  $B_{th1}$ is defined as the total
bandwidth obtained in Stages I and II (i.e., $B_{s}\alpha +B_{l}$)
such that $S(\pi)$ intersects with $D(\pi)$ at $\pi=0.468$. Here is
how the optimal pricing is determined:
\begin{itemize}
\item If $B_s\alpha+B_l\geq B_{th1}$ (e.g., $S_1(\pi)$ in Fig.
\ref{fig:inter_general}), the optimal price is $\pi^*=0.468$. The
total supply is no smaller (and often exceeds) the total demand.

\item If $B_s\alpha+B_l< B_{th1}$ (e.g.,
$S_2(\pi)$ in Fig. \ref{fig:inter_general}), the optimal price
occurs at the unique intersection point of $S(\pi)$ and $D(\pi)$
(where $D(\pi)$ has a negative first-order derivative). The total
supply equals total demand.
\end{itemize}


\item Now we are ready to show the threshold structure of the leasing decision.
\begin{itemize}
\item If the sensing result from Stage I satisfies $B_s\alpha\geq B_{th1}$, then the operator will not lease. This is because leasing will only increase the total cost without increasing the revenue, since the optimal price is fixed at $\pi^{\ast}=0.468$ and thus revenue is also fixed at $D(\pi^{\ast})$.

\begin{figure}[tt]
\centering
\includegraphics[width=0.3\textwidth]{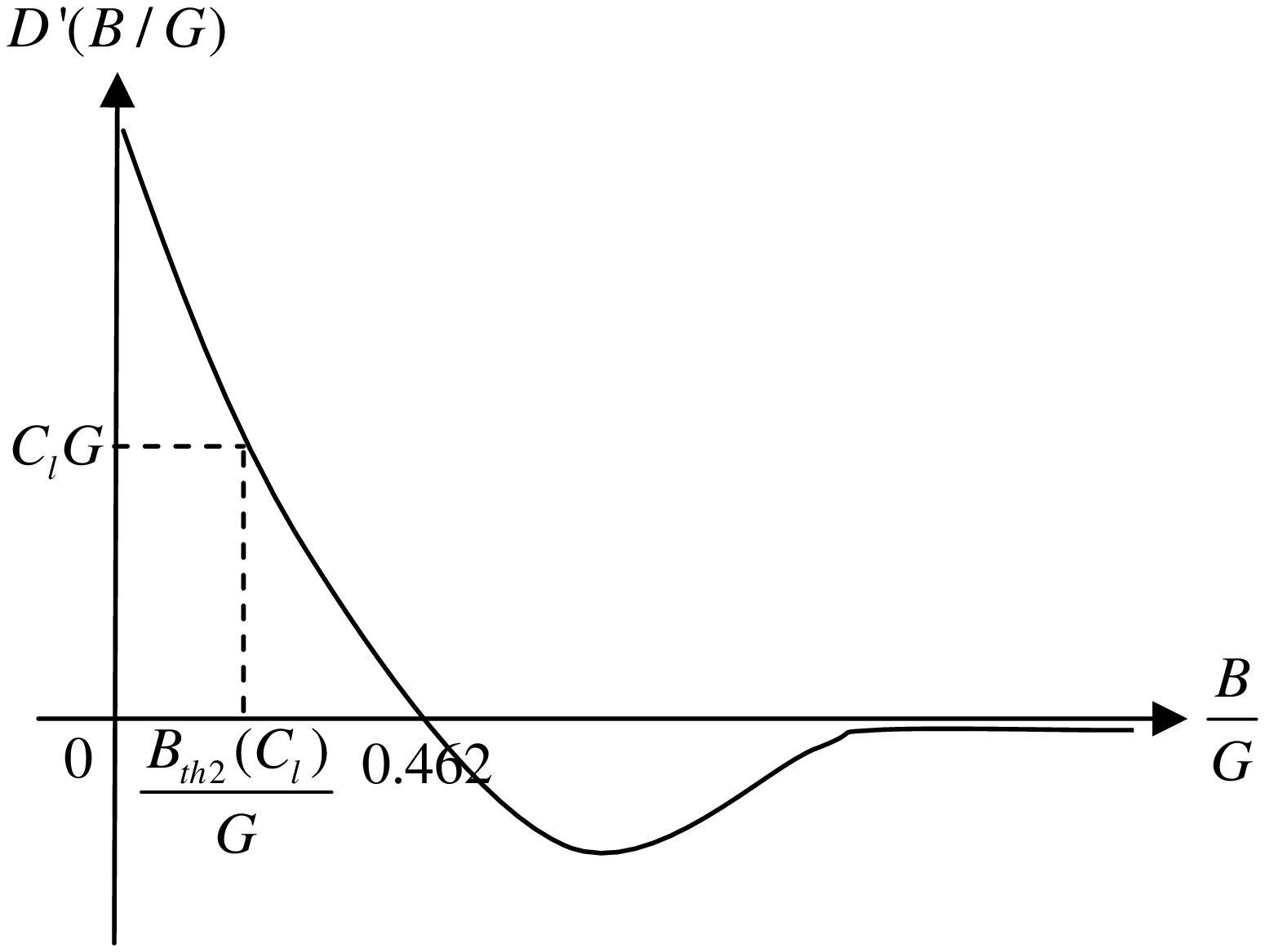}
\caption{The relation between the normalized total bandwidth $B/G$
and the derivative of the revenue $D'(B/G)$.} \label{fig:dDdB}
\end{figure}

\item Let us focus on the case where the sensing result from Stage I satisfies $B_s\alpha< B_{th1}$. Let us define $B=B_s\alpha+B_l$, then we have $B=G/Q(\pi)$
and $\pi=\ln(1+G/B)-G/(G+B)$. This enables us to rewrite  $D(\pi)$
as a function of total resource $B$ only,
$$D(B)=B\left[\ln\left(1+\frac{G}{B}\right)-\frac{G}{G+B}\right].$$
The first-order derivative of $D(B)$ is
\begin{equation}\label{eq:dD(B)}
D'(B)=\ln\left(1+\frac{1}{B/G}\right)-\frac{1}{1+B/G}-\frac{1}{(1+B/G)^2},\end{equation}
which denotes the increase of revenue $D(B)$ due to unit increase in
bandwidth $B$.
Since obtaining each unit bandwidth has a cost of $C_{l}$ in Stage
II,  the operator will only lease positive amount of bandwidth if
and only if $D'(B_{s}\alpha)>C_l$. To facilitate the discussions, we
will plot the function of $D'(B/G)$ in Fig.~\ref{fig:dDdB}, with the
understanding that $D'(B/G)=D'(B)G$. The intersection point of
$B/G=0.462$ in Fig.~\ref{fig:dDdB} corresponds to the point of
$\pi=0.468$ in Fig.~\ref{fig:inter_general}. The positive part of
$D'(B)$ on the left side of $B/G=0.462$ in Fig.~\ref{fig:dDdB}
corresponds to the part of $D(\pi)$ with a negative first-order
derivative in Fig.~\ref{fig:inter_general}.
%
%
For any value $C_l$, Fig.~\ref{fig:dDdB} shows that there exists a
unique threshold $B_{th2}(C_{l})$ such that
$D'(B_{th2}(C_{l})/G)=C_{l}G$, i.e., $D'(B_{th2}(C_{l}))=C_{l}$.
Then the optimal leasing amount will be $B_{th2}(C_{l})-B_{s}\alpha$
if the bandwidth obtained from sensing $B_{s}\alpha$ is less than
$B_{th2}(C_{l})$, otherwise it will be zero.
%
\end{itemize}

\end{itemize}

\subsection{Threshold structure of pricing and Observation \ref{ob3}}
Based on the proofs above, we show that Observation \ref{ob3} also
holds for the general case as follows. Let us denote the optimal
sensing decision as $B_s^*$, and consider two sensing realizations
$\alpha_1$ and $\alpha_2$ in time slots 1 and 2, respectively.
Without loss of generality, we assume that $\alpha_1<\alpha_2$.
\begin{itemize}
\item If $B_s^*\alpha_2\geq B_{th1}$, then the optimal price in time slot
2 is $\pi^*=0.468$ (see Fig.~\ref{fig:inter_general}). The optimal
price in time slot 1 is always no smaller than $0.468$.

\item If $B_s^*\alpha_1< B_s^*\alpha_2<B_{th1}$, then we need to consider three subcases:
\begin{itemize}
\item If $B_s^*\alpha_1< B_s^*\alpha_2\leq B_{th2}(C_l)$, then
the operator will lease up to the threshold in both time slots,
i.e., $B_l^*=B_{th2}(C_l)-B_s^*\alpha_1$ in time slot 1 and
$B_l^*=B_{th2}(C_l)-B_s^*\alpha_2$ in time slot 2. Then optimal
prices in both time slots are the same.

\item If $B_s^*\alpha_1\leq  B_{th2}(C_l)< B_s^*\alpha_2$, then the
operator  will lease $B_l^*=B_{th2}(C_l)-B_s^*\alpha_1$ in time slot
1 and will not lease in time slot 2. Thus the total bandwidth in
time slot 1 is smaller than that of time slot 2, and the optimal
price in time slot 1 is larger.

\item If $B_{th2}(C_l)\leq B_s^*\alpha_1< B_s^*\alpha_2$, then
the operator in both time slots will not lease and total bandwidth
in time slot 1 is smaller, and the optimal price in time slot 1 is
larger.
\end{itemize}
To summarize, the optimal price $\pi^*$ in Stage III is
non-increasing in $\alpha$. And the operator will charge a constant
price ($\pi^*=0.468$) to the users as long as the total bandwidth
obtained through sensing and leasing does not exceed the threshold
$B_{th2}(C_l)$.\hfill\rule{2mm}{2mm}
\end{itemize}

\bibliographystyle{ieeetran}
\bibliography{references_workingpaper}

\end{document}